\documentclass[12pt, doublespacing, onecolumn]{IEEEtran}
\makeatletter

\newcommand{\Rmnum}[1]{\expandafter\@slowromancap\romannumeral #1@}

\makeatother

\usepackage{amsmath,bm}
\usepackage{amssymb}
\usepackage{extarrows}
\usepackage{algorithm}
\usepackage{algorithmic}
\usepackage{graphicx} %use graph format
\usepackage{epstopdf}
\usepackage{amsthm}
\usepackage{enumerate}
\usepackage[T1]{fontenc}
\usepackage{cite,bm,color,url,textcomp}
\usepackage[hidelinks,breaklinks=true]{hyperref}
\newcommand\blfootnote[1]{%
  \begingroup
  \renewcommand\thefootnote{}\footnote{#1}%
  \addtocounter{footnote}{-1}%
  \endgroup
}

\usepackage{setspace,xspace}
\newcommand{\ba}{\begin{array}}
\newcommand{\ea}{\end{array}}
\newcommand{\be}{\begin{displaymath}}
\newcommand{\ee}{\end{displaymath}}
\newcommand{\ben}{\begin{equation}}  
\newcommand{\een}{\end{equation}}
\newcommand{\bea}{\begin{equation}\begin{aligned}}
\newcommand{\eea}{\end{aligned}\end{equation}}      
\newcommand{\bena}{\begin{eqnarray}}
\newcommand{\eena}{\end{eqnarray}}
\newcommand{\beqa}{\begin{eqnarray*}}
\newcommand{\enqa}{\end{eqnarray*}}
\newcommand{\f}{\frac}
\newcommand{\bc}{\begin{center}}
\newcommand{\ec}{\end{center}}
\newcommand{\bi}{\begin{itemize}}
\newcommand{\ei}{\end{itemize}}
\newcommand{\benu}{\begin{enumerate}}
\newcommand{\eenu}{\end{enumerate}}
\newcommand{\bdes}{\begin{description}}
\newcommand{\edes}{\end{description}}
\newcommand{\bt}{\begin{tabular}}
\newcommand{\et}{\end{tabular}}

\newcommand \abf{{\bf a}}
\newcommand \bbf{{\bf b}}
\newcommand \cbf{{\bf c}}
\newcommand \dbf{{\bf d}}

\newcommand \ibf{{\bf i}}
\newcommand \jbf{{\bf j}}
\newcommand \kbf{{\bf k}}

\newcommand \nbf{{\bf n}}

\newcommand \sbf{{\bf s}}
\newcommand \tbf{{\bf t}}

\newcommand \zbf{{\bf z}}

\newcommand \Abf{{\bf A}}
\newcommand \Bbf{{\bf B}}
\newcommand \Cbf{{\bf C}}
\newcommand \Dbf{{\bf D}}
\newcommand \Ebf{{\bf E}}
\newcommand \Fbf{{\bf F}}
\newcommand \Gbf{{\bf G}}
\newcommand \Hbf{{\bf H}}
\newcommand \Ibf{{\bf I}}
\newcommand \Jbf{{\bf J}}

\newcommand \Pbf{{\bf P}}
\newcommand \Qbf{{\bf Q}}
\newcommand \Rbf{{\bf R}}
\newcommand \Sbf{{\bf S}}

\newcommand \Xbf{{\bf X}}
\newcommand \Ybf{{\bf Y}}

\newtheorem{theorem}{Theorem}[section]

\newcommand{\Rnum}{{\mathbb R}}
\newcommand{\Cnum}{{\mathbb C}}

\newcommand{\Acal}{{\cal A}}
\newcommand{\Bcal}{{\cal B}}
\newcommand{\Ccal}{{\cal C}}

\newcommand{\circlambda}{\mbox{$\Lambda$
             \kern-.85em\raise1.5ex
             \hbox{$\scriptstyle{\circ}$}}\,}

\newtheorem{Theorem}{Theorem}[section]
\newtheorem{Definition}[Theorem]{Definition}
\newtheorem{Proposition}[Theorem]{Proposition}

\newtheorem{Corollary}[Theorem]{Corollary}

\theoremstyle{plain}
\newtheorem{lemma}{Lemma}
\newtheorem*{remark}{Remark}
\newtheorem{property}[Theorem]{Property}
\makeatletter
\def\thmhead@plain#1#2#3{%
  \thmname{#1}\thmnumber{\@ifnotempty{#1}{ }\@upn{#2}}%
  \thmnote{ {\the\thm@notefont#3}}}
\let\thmhead\thmhead@plain
\makeatother
\begin{document}

\title{Constructions of Polyphase Golay Complementary Arrays}

\author{Cheng Du, \textit{Student Member, IEEE,} \  Yi Jiang, \textit{Member, IEEE,}    % <-this % stops a space
%Key Laboratory for Information Science of Electromagnetic Waves (MoE) \\
%Shanghai Institute for Advanced Communication and Data Science \\
%Dept. of Communication Science $\&$ Engineering, \\Fudan University, Shanghai, China
%Emails:17210720089@fudan.edu.cn, yijiang@fudan.edu.cn, xwang11@fudan.edu.cn
}
\maketitle
\blfootnote{The work was supported by National Natural Science Foundation of China Grant No. 61771005. ({\em Corresponding author: Yi Jiang})

C. Du and Y. Jiang are Key Laboratory for Information Science of Electromagnetic Waves (MoE), Department of Communication Science and Engineering, School of Information Science and Technology, Fudan University, Shanghai, China (E-mails: cdu15@fudan.edu.cn, yijiang@fudan.edu.cn).}
\begin{abstract}
  Golay complementary matrices (GCM) have recently drawn considerable attentions owing to its potential applications in omnidirectional precoding. In this paper we generalize the GCM to multi-dimensional Golay complementary arrays (GCA) and propose new constructions of GCA pairs and GCA quads. These constructions are facilitated by introducing a set of identities over a commutative ring. We prove that a quaternary GCA pair is feasible if the product of the array sizes in all dimensions is a quaternary Golay number with an additional constraint on the factorization of the product. For the binary GCM quads, we conjecture that the feasible sizes are arbitrary, and verify for sizes within $78\times 78$ and other less densely distributed sizes. For the quaternary GCM quads, all the positive integers within 1000 can be covered for the size in one dimension.
\end{abstract}


% \fi
\IEEEpeerreviewmaketitle

\section{Introduction}
%\IEEEPARstart
The binary Golay sequence pair with entries $\{1, -1\}$ was first introduced by Golay \cite{golay1961complementary}\cite{golay1962note}, whose respective autocorrelation functions add to be a $\delta$-function. Then Turyn improved the recursive construction proposed by Golay, to generate the binary Golay sequence pairs of length $2^a10^b26^c, a, b, c\geq 0$ \cite{turyn1974hadamard}. The above lengths are referred as the binary Golay numbers. In \cite{borwein2004complete} it was verified by exhaustive computational search that the lengths above cover all the 14 feasible numbers within 100.

To obtain more feasible lengths, the binary Golay sequence set was first considered in \cite{tseng1972complementary}, which allows for $L\ge 2$ sequences whose respective autocorrelation functions add to be a $\delta$-function. It was conjectured that the binary Golay sequence quads, i.e., for $L=4$, exist for arbitrary lengths \cite{turyn1974hadamard}, which plays an important role in construction of the Hadamard matrices \cite{goethals1970skew}. Despite of the promising existence pattern, it is quite challenging to find a construction method for arbitrary lengths. 

Another generalization is the polyphase Golay sequence pair as discussed in \cite{sivaswamy1978multiphase, frank1980polyphase, craigen2002complex}, whose entries are the $N$-th unit roots where $N\geq 2$. For the quaternary Golay sequence pair whose entries are $\{1, -1, j, -j\}$ where $j$ is the imaginary unit, the known feasible lengths revealed by recursive constructions, referred as the quaternary Golay numbers, are $2^{a+u}3^{b}5^{c}11^{d}13^{e}$ where $a, b, c, d, e, u \geq 0$, $b+c+d+e \leq a+2u+1$, $u \leq c+e$ \cite{craigen2002complex}. This existence pattern is much denser than the binary counterpart. And recently, it has been verified via exhaustive computational search that the quaternary Golay numbers cover all the $17$ feasible lengths within $28$ \cite{bright2021complex}.

A multi-dimensional generalization named the binary Golay complementary array (GCA) pair, i.e., a pair of tensors whose respective multi-dimensional autocorrelation functions add to be a multi-dimensional $\delta$-function, has been studied in \cite{luke1985sets, dymond1992barker, jedwab2007golay, fiedler2008multi, parker2011generalised}. The work above focused on both the existence pattern and the enumeration of the arrays. The known existence pattern of the binary GCA pair is similar to the one-dimensional counterpart: the array size in each dimension is a binary Golay number \cite{dymond1992barker}.

The ultimate generalization is the polyphase GCA set \cite{parker2011generalised,jiang2019autocorrelation,li2021construction,avis20103, avis2021three}. Specifically, we have constructed the polyphase Golay complementary matrix (GCM) set recently and found an application in the omnidirectional transmission by a massive multi-input multi-output (MIMO) antenna array \cite{jiang2019autocorrelation,li2021construction,girnyk2021efficient}. The feasible sizes of the quaternary GCM pairs are $2g_q^{(1)}\times g_q^{(2)}$ or $g_q^{(1)}\times 2g_q^{(2)}$ \cite{li2021construction}, where $g_q^{(1)}$ and $g_q^{(2)}$ are two quaternary Golay numbers. The feasible sizes of the quaternary GCM quads are $g_q^{(1)}\times g_q^{(2)}$ \cite{li2021construction}. Besides, a denser existence pattern of the ternary GCA triads, whose entries are $\{1, -\frac{1}{2}+\frac{\sqrt{3}}{2}j, -\frac{1}{2}-\frac{\sqrt{3}}{2}j\}$, has been discovered in \cite{avis20103}\cite{avis2021three}. Since an antenna array and an GCM shall be of the same size, it would be interesting to discover more feasible sizes of GCM to accommodate for antenna arrays of flexible sizes.

Another important application of the GCA is the precoding of orthogonal frequency division multiplexing (OFDM) signal to reduce its peak-to-average power ratio (PAPR) \cite{davis1999peak}\cite{paterson2000generalized}. However, the code rate is restricted by the total number of different Golay sequences obtained by projecting the high-dimensional arrays \cite{fiedler2008multi}. The work of enumeration of GCA of some special sizes have been extensively studied in \cite{davis1999peak,paterson2000generalized,schmidt2007complementary,fiedler2008multi,wang2021new}. 

In this paper, we focus on the existence pattern of the polyphase GCA set. By first establishing some identities over a commutative ring for recursive constructions, we construct the polyphase GCA pairs and GCA quads of more feasible sizes. Our contributions in this paper is three-fold. First, based on Lemma \ref{lem:2Identity}, we propose the construction of the quaternary GCA pairs in Theorem \ref{thm:G-N-ary} and derive their existence pattern in Corollary \ref{Qsize}. Second, based on Lemma \ref{lem:4Identity}, we construct the GCA quads in Theorem \ref{thm:quad} and specifically cover all the integers within $78\times 78$ for the feasible sizes of the binary GCM quads. Third, we propose Lemma \ref{lem:compromise}, which overcomes the cumbersome applicable conditions of Lemma \ref{lem:4Identity} for constructing the GCA quads. Based on Lemma \ref{lem:compromise}, we propose another construction of the GCA quads in Theorem \ref{thm:compromise}, and together with Theorem \ref{thm:quad}, we cover all the positive integers within $1000$ for the feasible sizes in one dimension of the quaternary GCA quads.

The remainder of this paper is organized as follows.

In Section \ref{SEC1}, we introduce three types of definitions of the GCA set, and provide an interpretation in the sight of commutative ring and its involutive automorphism. In Section \ref{SEC2}, we present four lemmas on the identities over a commutative ring, which is the cornerstone of the construction in the paper. And as an example to explain how these identities can be exploited to construct the GCA set, we provide a proof of the construction of the binary GCA pairs in \cite{dymond1992barker} using Lemma \ref{lem:2Identity}. 
In Section \ref{SEC3}, first we summarize the construction of the polyphase GCM pairs in \cite{li2021construction} and generalize it slightly in Theorem \ref{thm:N-ary}. Second, inspired by the idea in \cite{craigen2002complex}, we further propose Theorem \ref{thm:G-N-ary} as a generalization of Theorem \ref{thm:N-ary}. The feasible sizes are derived in Corollary \ref{Qsize}.
In Section \ref{SEC4}, first we summarize the construction of the polyphase GCM quads \cite{jiang2019autocorrelation} in Theorem \ref{thm:naive_cross}, which is based on the identity in Lemma \ref{lem:mnIdentity}. From this perspective, we generalize the construction slightly in Theorem \ref{thm:cross}. Second, based on the identities in Lemma \ref{lem:4Identity} and Lemma \ref{lem:compromise}, we give some new constructions of the polyphase GCA quads, generating much more feasible sizes.
The conclusions are given in Section \ref{SEC5}.

\section{Preliminaries: Polynomials of Golay Complementary Arrays and Commutative Ring} \label{SEC1}
The Golay complementary array (GCA) set has three types of equivalent definitions \cite{jedwab2007golay}\cite{parker2011generalised}. We summarize them in Definition \ref{def:GCA_1}, \ref{def:GCA_2} and \ref{def:GCA_3}, and the last two definitions based on polynomial will be adopted in this paper. Furthermore, a conception of commutative ring is introduced to depict the polynomial definition briefly.

For an $r$-dimensional complex-valued array $\Abf$ of size $s_1\times \cdots\times s_r$, its aperiodic autocorrelation function $\Rbf_A$ is defined as an array of size $(2s_1-1)\times \cdots\times (2s_r-1)$ with entries \cite{jedwab2007golay}:
\bea \label{corr}
& \Rbf_{A}\left[\delta_1, \cdots, \delta_r\right]\\
=\ &\sum_{i_1}\cdots\sum_{i_r} \Abf \left[i_1, \cdots, i_r\right] \overline{\Abf\left[i_1-\delta_1, \cdots, i_r-\delta_r\right]},
\eea
where $i_1,\cdots, i_r$ and $\delta_1, \cdots, \delta_r$ are array indices, $\Abf \left[i_1,\cdots, i_r\right] = 0$ if $i_k < 0$ or $i_k \ge s_k$ for any $k \in \left\{1, 2, \cdots, r\right\}$, and the overbar represents the complex conjugation. The array $\Abf$ is indexed from $\left[0, \cdots, 0\right]$, while the index $\left[0, \cdots, 0\right]$ corresponds to the center of the array $\Rbf_A$.

The weight of an array $\Abf$ is a nonnegative real number defined as \cite{craigen2002complex}
\ben \label{weight}
w(\Abf) = \sum_{i_1}\cdots\sum_{i_r} {\lvert \Abf \left[i_1, \cdots, i_r\right] \rvert} ^2. 
\een
\begin{Definition} \label{def:GCA_1}
A set of arrays $\{\Abf_1, \Abf_2, \cdots, \Abf_L\}$ with unimodular or zero entries is called a Golay complementary array (GCA) set if
\ben \label{ACondition}
\sum_{i=1}^{L}\Rbf_{A_i} = \sum_{i=1}^{L} w(\Abf_i) \cdot \bf{\Delta},
\een
where ${\bf{\Delta}} \in \Cnum^{(2s_1-1)\times \cdots\times (2s_r-1)}$ is an $r$-dimensional unit pulse function, i.e, the entries of $\bf{\Delta}$ are zeros that the center ${\bf{\Delta}}[0, \cdots, 0] = 1$.
\end{Definition}

When the entries are unimodular, the weight is exactly $\prod_{i=1}^r s_i$. When zero occurs in an array, we say the array is weight-deficient. When the entries are the $N$-th unit roots, we refer it as a polyphase array. Specifically, a binary array is one with entries $\{1, -1\}$ and a quaternary array is one with entries $\{1, -1, j, -j\}$ where $j=\sqrt{-1}$. An array is trivial in the $i$-th dimension if $s_i=1$, and when it's trivial in all dimensions, we say the array is trivial.

A GCA set degenerates into a Golay sequence set \cite{tseng1972complementary} or a Golay complementary matrix (GCM) set \cite{li2021construction} when $r=1$ or $r=2$ respectively. An example of a quaternary GCM pair of size $2\times 3$ is given as follows:

\ben
\Abf_1 = \begin{bmatrix}
  1&1&-1\\ -1&-j&-1
\end{bmatrix}, \quad
\Abf_2 = \begin{bmatrix}
  -1&-1&1\\ -1&-j&-1
\end{bmatrix}.
\een
And their autocorrelation functions are
\bea
\Rbf_{A_1} &= \begin{bmatrix}
  -1&-1+j&j&-1-j&1\\ 0&0&6&0&0\\ 1&-1+j&-j&-1-j&-1
\end{bmatrix}, \\
\Rbf_{A_2} &= \begin{bmatrix}
  1&1-j&-j&1+j&-1\\ 0&0&6&0&0\\ -1&1-j&j&1+j&1
\end{bmatrix},
\eea
whose summation is a two-dimensional delta-function.

The multi-variable polynomial $A(\zbf)$ of an array $\Abf$ is defined as \cite{jedwab2007golay}
\bea \label{polynomial}
& A(\zbf) := A\left(z_1, \cdots, z_r\right) \\
=\ &\sum_{i_1}\cdots\sum_{i_r} \Abf \left[i_1, \cdots, i_r\right] z_1^{i_1}\cdots z_r^{i_r},
\eea
where $z_1, \cdots, z_r$ are indeterminates. Here we use the notation $(\cdot)$ rather than $[\cdot]$ to distinguish the indeterminates from the array indices. 

In the context of signal processing, the polynomial of an array is the Z-transform of the signal. Then by the Wiener–Khinchin theorem, the Z-transform of the autocorrelation function of an array is equal to the power spectrum of the array, i.e.,
\ben \label{weiner_khinchin}
R_A(\zbf) = A(\zbf)\overline{A(\zbf^{-1})},
\een
where 
\ben \label{eqAzbar}
\overline{A(\zbf^{-1})} = \sum_{i_1}\cdots\sum_{i_r} \overline{\Abf \left[i_1, \cdots, i_r\right]} z_1^{-i_1}\cdots z_r^{-i_r}.
\een
Since the Z-transform of the $\delta$-function equals $1$, combining (\ref{ACondition}) and (\ref{weiner_khinchin}), it's straightforward to give the second definition of GCA set \cite{parker2011generalised}:
\begin{Definition} \label{def:GCA_2}
A set of arrays $\{\Abf_1, \Abf_2, \cdots, \Abf_L\}$ with unimodular or zero entries is called a GCA set if 
\ben \label{PPCondition}
\sum_{i=1}^{L}A_i(\zbf) \overline{A_i(\zbf^{-1})} = \sum_{i=1}^{L} w(\Abf_i).
\een
\end{Definition}
% In brief, the complementarity lies in the frequency domain.

Nevertheless, the negative power of $z_i$ in the expansion of $\overline{A(\zbf^{-1})}$ makes it inconvenient to correspond with the Z-transform of some array. The polynomial of the flipped and conjugated version of $\Abf$ is defined as
\bea \label{eqastr}
& A^*(\zbf) := A^*\left(z_1, \cdots, z_r\right) = \sum_{i_1}\cdots\sum_{i_r} \\
&\overline{\Abf \left[s_1-1-i_1, \cdots, s_r-1-i_r\right]} z_1^{i_1}\cdots z_r^{i_r}.
\eea
where $*$ denotes rotating an array $180$ degree in each dimension and conjugating it. Comparing (\ref{eqAzbar}) and (\ref{eqastr}), we see that 
\ben \label{ppp}
A^*(\zbf) = z_1^{s_1-1}\cdots z_r^{s_r-1} \overline{A(\zbf^{-1})}.
\een
Then by (\ref{PPCondition}) and (\ref{ppp}), we have the third definition \cite{parker2011generalised}:
\begin{Definition} \label{def:GCA_3}
A set of arrays $\{\Abf_1, \Abf_2, \cdots, \Abf_L\}$ with unimodular or zero entries is called a GCA set if
\ben \label{PCondition}
\sum_{i=1}^{L}A_i(\zbf)A_i^*(\zbf) = \sum_{i=1}^{L} w(\Abf_i) z_1^{s_1-1} \cdots z_r^{s_r-1}.
\een  
\end{Definition}

A more basic interpretation of Definition \ref{def:GCA_3} is based on the concept of commutative ring and its involutive automorphism (the following introduction of these concepts, properties and their proofs can be found in any textbook of abstract algebra, e.g.,  \cite{jacobson2012basic}).

\begin{Definition} \label{commutative_ring}
  A commutative ring is a triple $(\Rnum, +, \cdot)$, where $\Rnum$ is a non-vacuous set(do not confuse it with the field of real number), $+$ and $\cdot$ are two compositions for elements in $\Rnum$, satisfying the following conditions:
  \begin{enumerate}
    \item Closeness. The set $\Rnum$ is closed under the compositions $+$ and $\cdot$, i.e., $a+b\in \Rnum$ and $a\cdot b \in \Rnum$ hold for any $a, b\in \Rnum$.
    \item Associativity. The compositions $+$ and $\cdot$ satisfy the associative law, i.e.,  for any $a, b, c\in \Rnum$, $(a+b)+c=a+(b+c)$ and $(a\cdot b)\cdot c=a\cdot(b\cdot c)$.
    \item Commutativity. The compositions $+$ and $\cdot$ satisfy the commutative law, i.e.,  for any $a, b\in \Rnum$, $a+b=b+a$ and $a\cdot b=b\cdot a$.
    \item Distributivity. The compositions $+$ and $\cdot$ satisfy the distributive law, i.e.,  for any $a, b, c\in \Rnum$, $a\cdot(b+c)=a\cdot b+a\cdot c$ and $(b+c)\cdot a=b\cdot a+c\cdot a$.
    \item Identity elements. There exist an element $0\in \Rnum$ and an element $1\in \Rnum$, such that for any $a\in \Rnum$, $a+0=0+a=a$ and $a\cdot 1=1\cdot a=a$.
    \item Invertibility. Any element in $\Rnum$ is invertible in the sense of the composition $+$, i.e., for any $a\in \Rnum$, there exists an element in $ \Rnum$ denoted by $a^{-}$, such that $a+a^{-}=a^{-}+a=0$. 
  \end{enumerate}
\end{Definition}
For ease of notation, $b+a^{-}$ is written as $b-a$ and $a\cdot b$ is written as $ab$.
\begin{property} \label{prop:inv}
For any $a,b \in {\mathbb R}$, it holds that 
\ben \label{invert}
ab^- = (ab)^- = a^-b.
\een
\end{property}
\begin{proof}
  Noting that $a0=a(0+0)=a0+a0$, addition of $(a0)^-$ gives $a0=0$. Then $ab^- + ab = a(b+b^-) =a0=0$; similarly, $ab+ab^-=0$. Thus $ab^-=(ab)^-$. Similarly, $a^-b=(ab)^-$.
\end{proof}

Polynomial ring is an example of commutative ring. Let $\Rnum = \{A(\zbf)\}$ be the set of the multi-variable polynomials defined in (\ref{polynomial}). The composition $+$ and $\cdot$ are defined as the conventional polynomial summation and multiplication respectively. The identity elements $0$ and $1$ are the integers $0$ and $1$ respectively. It's straightforward to verify that the polynomial ring defined above satisfy the six conditions in Definition \ref{commutative_ring}.

It is well-known that the multiplication of two polynomials amounts to the convolution of two arrays. Relating the compositions $+$ and $\cdot$ of a commutative ring of complex arrays to the element-wise summation and the convolution of two arrays respectively, we see that the array ring and the polynomial ring are isomorphic \cite{jacobson2012basic}. Roughly speaking, they have the same structure or they are identical. This explains the equivalence of Definition \ref{def:GCA_1} and Definition \ref{def:GCA_3}.

Another important concept is the involutive automorphism of a commutative ring. 
\begin{Definition} \label{involutive_automorphism}
  An involutive automorphism of a commutative ring $\Rnum$ is a bijective map $*$ of $\Rnum$ into itself, which map an element $a\in \Rnum$ to its image $a^* \in \Rnum$, such that for any $a, b\in \Rnum$, 
  \bea \label{p_inv}
  (a+b)^* &= a^*+b^*, \quad
  (a\cdot b)^* = a^* \cdot b^*, \\
  1^* &= 1, \qquad \qquad \quad
  (a^*)^* = a.
  \eea
\end{Definition}
\begin{property} \label{prop:involution}
For the involutive automorphism, it holds that
\ben \label{invert_involute}
(a^-)^* = (a^*)^-.
\een
\end{property}
\begin{proof}
  Noting that $0^* = (0+0)^* = 0^* + 0^*$, addition of $(0^*)^-$ gives $0 = 0^*$, then $(a^-)^* + a^* = (a^- + a)^* = 0^* = 0$. Similarly, $a^*+(a^-)^*=0$. Thus $(a^-)^* = (a^*)^-$.
\end{proof}

It's straightforward to verify that the map from $A(\zbf)$ to $A^*(\zbf)$ defined in (\ref{eqastr}) is an involutive automorphism. Thus (\ref{PCondition}) can be rewritten as 
\ben
\sum_{i=1}^{L} a_{i}a_{i}^* = x,
\een
where $\{a_1, \cdots, a_L\} \subset \Rnum$ is the desired subset, and $x = \sum_{i=1}^{L} w(\Abf_i)\prod_{i=1}^r z_i^{s_i-1}$ is a specific element in $\Rnum$.
\begin{remark}
  By the concept of commutative ring and its involutive automorphism, we treat a multi-variable polynomial as an entity instead of a tedious expansion, so that we can summarize and exploit their property briefly. This provides a basic methodology in the recursive constructions of GCAs in the paper.
\end{remark}

\section{Four Lemmas} \label{SEC2}
This section presents four lemmas on the identities over a commutative ring, which will be the cornerstones of the construction of the polyphase GCA set developed later in the paper. 

The following Lemma \ref{lem:2Identity} is essentially a restatement of  \cite[Theorem 1]{parker2011generalised} in the framework of commutative ring. %   are well-known in the construction of the polyphase Golay sequence pair and quad respectively. We extract Lemma \ref{lem:mnIdentity} from \cite[Lemma \MakeUppercase{\romannumeral 3}.3]{jiang2019autocorrelation} and propose Lemma \ref{lem:compromise} as a compromise between Lemma \ref{lem:2Identity} and Lemma \ref{lem:mnIdentity}.

 \begin{lemma}\label{lem:2Identity}
    Given a commutative ring $\Rnum$ with an involutive automorphism $*$, and $a, b, c, d \in \Rnum$, suppose
    \ben \label{quaternion}
    e = ac+bd, \quad f = b^*c-a^*d,
    \een
    then
    \ben \label{norm}
    ee^* + ff^* = \left(aa^*+bb^*\right)\left(cc^*+dd^*\right).
    \een
  \end{lemma}

  \begin{proof}
    By the property of a commutative ring with an involutive automorphism, we have 
    \bea
    &ee^*+ff^* \\
    =\ &(ac+bd)(ac+bd)^*+(b^*c-a^*d)(b^*c-a^*d)^*\\
    \overset{(\Acal)}{=}& (ac+bd)(a^*c^*+b^*d^*)+(b^*c-a^*d)(bc^*-ad^*)\\
    \overset{(\Bcal)}{=}& ac(a^*c^*)+ac(b^*d^*)+bd(a^*c^*)+bd(b^*d^*)+\\
    &b^*c(bc^*)+(b^*c)(ad^*)^- + (a^*d)^-(bc^*) + (a^*d)^-(ad^*)^-\\
    \overset{(\Ccal)}{=}& aa^*cc^* + bb^*cc^* + aa^*dd^* + bb^*dd^* + \\
    &ab^*cd^* + a^*bc^*d - ab^*cd^* - a^*bc^*d \\
    =& \left(aa^*+bb^*\right)\left(cc^*+dd^*\right)
    \eea
    where $\overset{(\Acal)}{=}$ follows from Definition \ref{involutive_automorphism} and Property \ref{prop:involution}, $\overset{(\Bcal)}{=}$ follows from the distributivity, $\overset{(\Ccal)}{=}$ follows from the commutativity, the associativity and Property \ref{prop:inv}.
  \end{proof}
  
%   Noting that by replacing $b$ with $b^*$, then (\ref{identity}) has a matrix form:
%   \ben
%   \begin{bmatrix}
%     e&f\\ f^*&-e^*
%   \end{bmatrix}=
%   \begin{bmatrix}
%     c&d\\ d^*&-c^*
%   \end{bmatrix}
%   \begin{bmatrix}
%     a&b\\ b^*&-a^*
%   \end{bmatrix}
%   \een
%   where each matrix is unitary in the sense that
%   \bea
%   &\begin{bmatrix}
%     a&b\\ b^*&-a^*
%   \end{bmatrix}
%   \begin{bmatrix}
%     a&b\\ b^*&-a^*
%   \end{bmatrix}^* = 
%   \begin{bmatrix}
%     a&b\\ b^*&-a^*
%   \end{bmatrix}^*
%   \begin{bmatrix}
%     a&b\\ b^*&-a^*
%   \end{bmatrix}\\
%   &= (aa^*+bb^*)
%   \begin{bmatrix}
%     1&0\\ 0&1
%   \end{bmatrix}
%   \eea
%   where $[\cdot]^*$ denotes transposing the matrix and replacing each element with its image. 
  A matrix form of Lemma \ref{lem:2Identity} has been employed in \cite{eliahou1991golay,parker2011generalised,budivsin2014paraunitary,budivsin2018paraunitary} to construct and enumerate the Golay sequence pairs and the Golay array pairs.
  
 The following  Lemma \ref{lem:4Identity} is a reproduction of the so-called Lagrange identity \cite[Theorem L]{yang1989composition}.
  \begin{lemma} \label{lem:4Identity}
  Given a commutative ring $\Rnum$ with an involutive automorphism *, and $a, b, c, d, e, f, g, h \in \Rnum$, suppose
    \bea \label{4Identity}
    p &= af^*-b^*e+cg+dh,\\
    q &= a^*e+bf^*-ch^*+dg^*,\\
    r &= c^*e-df+ah^*+bg,\\
    s &= -cf-d^*e+ag^*-bh,
    \eea
    then
    \bea \label{4Lagrange}
    &pp^*+qq^*+rr^*+ss^*\\
    =\ & \left(aa^*+bb^*+cc^*+dd^*\right)\left(ee^*+ff^*+gg^*+hh^*\right).
    \eea
\end{lemma}
%It's worth noting that by dividing the right side of (\ref{4Identity}) into left half part and right half part, 8 parts is obtained, each of which is identical in sense of (\ref{norm}) to corresponding item in (\ref{compromise}). It can be verified directly that the sum of all the cross-product items of adjacent parts in expansion of $pp^*+qq^*+rr^*+ss^*$ is 0. 
The verification of  Lemma \ref{lem:4Identity} is more laborious than Lemma \ref{lem:2Identity} but is straightforward.
\begin{remark}
  % (\ref{4Lagrange}) is named Lagrange identity \cite[Theorem L]{yang1989composition} and 
According to \cite{baez2002octonions}, the pair $(e, f)$ in Lemma \ref{lem:2Identity} and the quad $(p, q, r, s)$ can be viewed as the quaternion and the octonion respectively, which are normed division algebras constructed by the well-known Cayley-Dickson process. This explains more elegantly why the identities hold. Unfortunately, there does not exist a Lagrange identity of eight components since the real numbers, complex numbers, quaternions and octonions are the only normed division algebras \cite[Theorem 1]{baez2002octonions}. % Thus for constructing GCA sets of larger cardinality, one shall consider a combination of Lemma \ref{lem:4Identity} and Lemma \ref{lem:mnIdentity}.
\end{remark}

  \begin{lemma} \label{lem:mnIdentity}
  Given a commutative ring $\Rnum$ with an involutive automorphism *, and $a_1, \cdots a_m, b_1, \cdots, b_n \in \Rnum$, suppose
  \ben \label{eq:cross}
  c_{ij} = a_i b_j, \quad 1\leq i\leq m, 1\leq j\leq n,
  \een
  then
  \ben \label{cross_item}
  \sum_{i=1}^{m}\sum_{j=1}^{n} c_{ij}c_{ij}^{*} = \left(\sum_{i=1}^{m} a_i a_i^*\right) \left(\sum_{j=1}^{n} b_j b_j^*\right).
  \een
\end{lemma}

Lemma \ref{lem:mnIdentity} is a more abstract version of \cite[Lemma \MakeUppercase{\romannumeral 3}.3]{jiang2019autocorrelation}. With the notion of commutative ring, the proof of the lemma is much more straightforward than that of \cite[Lemma \MakeUppercase{\romannumeral 3}.3]{jiang2019autocorrelation}. % as done in the proof of Lemma \ref{lem:2Identity}.

Combining Lemma \ref{lem:2Identity} with Lemma \ref{lem:mnIdentity}, we propose the following identity. %, which makes a compromise between the cardinality of the set and the size of arrays as shown in Section \ref{SEC5}.
\begin{lemma} \label{lem:compromise}
  Given a commutative ring $\Rnum$ with an involutive automorphism *, and $a_1,\cdots a_m, b_1, \cdots, b_{n} \in \Rnum$, $m, n$ are even numbers. For $1\leq i\leq \frac{m}{2}, 1\leq j\leq \frac{n}{2}$, suppose
  \bea \label{compromise}
  c_{ij} &= a_{2i-1} b_{2j-1} + a_{2i}b_{2j},\\
  d_{ij} &= a_{2i-1} b_{2j}^* - a_{2i}b_{2j-1}^*,
  \eea
  then
  \ben \label{eq.lem4}
  \sum_{i=1}^{\frac{1}{2}m}\sum_{j=1}^{\frac{1}{2}n} \left(c_{ij}c_{ij}^{*} + d_{ij}d_{ij}^{*}\right) = \left(\sum_{i=1}^{m} a_i a_i^*\right) \left(\sum_{j=1}^{n} b_j b_j^*\right).
  \een
\end{lemma}
\proof 
According to Lemma \ref{lem:2Identity}, 
\ben 
c_{ij}c_{ij}^{*} + d_{ij}d_{ij}^{*} = (a_{2i-1} a_{2i-1}^* + a_{2i}a_{2i}^*)(b_{2j-1} b_{2j-1}^* + b_{2j}b_{2j}^*), \nonumber
\een 
from which (\ref{eq.lem4}) follows.
\endproof 
% \section{Existing Constructions of GCA Pairs, Interpretations and Improvement} \label{SEC2}
%   From some binary Golay sequence pairs of length $2, 10, 26$ discovered by exhaustive computational search \cite{golay1961complementary, golay1962note}, \cite{turyn1974hadamard} proposed a recursive method to construct the binary Golay sequence pairs of length $2^a10^b26^c$ where $a, b, c\geq 0$ (referred as a binary Golay number). In \cite{dymond1992barker}, Dymond generalized the method to construct the binary GCA pairs of size $s_1\times \cdots \times s_r$ where $s_i$ is a binary Golay number for any $i=1, \cdots, r$. We list it here with a brief proof, as an example to show how the definitions in Section \ref{SEC1} can be employed.
To illustrate how the identities established in the above four lemmas can be employed in the construction of the GCA sets, we provide a new and more succinct proof of the well-known construction of the binary GCA pair \cite{dymond1992barker} using Lemma \ref{lem:2Identity}.
  \begin{theorem}[\cite{dymond1992barker}] \label{thm:binary}
    Given a binary GCA pair $\{\Abf, \Bbf\}$ of size $s_1\times s_2\times \cdots\times s_r$, and another one $\{\Cbf, \Dbf\}$ of size $t_1\times t_2\times \cdots\times t_r$, then a binary GCA pair $\{\Ebf, \Fbf\}$ of size $s_1t_1\times s_2t_2\times \cdots\times s_rt_r$ can be constructed as
    \bea \label{dymond}
    \Ebf &= \f{1}{2} \left[\Abf \otimes \left(\Cbf+\Dbf\right)+\Bbf \otimes \left(\Cbf-\Dbf\right)\right],\\
    \Fbf &= \f{1}{2} \left[\Bbf^* \otimes \left(\Cbf+\Dbf\right)-\Abf^* \otimes \left(\Cbf-\Dbf\right)\right]
    \eea
    where $*$ denotes flipping and conjugating an array in all dimensions, $\otimes$ denotes the Kronecker product.
  \end{theorem}
%   Practically, the binary Golay sequence pairs of length $2, 10, 26$ searched by computer are viewed as arrays trivial in all but one dimension, and fed into Theorem \ref{thm:binary} to construct the binary GCA pairs.

  To prove Theorem \ref{thm:binary},  we first establish the connection between the Kronecker product and the polynomial multiplication, and then assign specific polynomials to $a, b, c, d$ in Lemma \ref{lem:2Identity}.
  \begin{Proposition} \label{prop:kron}
    Given an array $\Abf$ of size  $s_1\times \cdots\times s_r$, and another array $\Bbf$ of size $t_1\times \cdots\times t_r$, suppose $\Gbf = \Abf \otimes \Bbf$, then $G(\zbf) = A(\zbf^{\tbf})B(\zbf)$ where $A(\zbf^\tbf)$ is an abbreviation of $A(z_1^{t_1}, \cdots, z_r^{t_r})$.
  \end{Proposition} 
  \begin{proof}
    By the property of Kronecker product, we have
    \ben
    \Gbf\left[k_1,\cdots, k_r\right] = \Abf\left[i_1, \cdots, i_r\right] \Bbf\left[j_1, \cdots, j_r\right]
    \een
    where $k_l = i_l t_l + j_l$ for any $l\in\{1,2, \cdots, r\}$, then 
    \bea
    &G(\zbf) \\
    =\ & \sum_{i_1, j_1}\cdots \sum_{i_r, j_r} \Abf\left[i_1, \cdots, i_r\right] \Bbf\left[j_1, \cdots, j_r\right] \prod_{l=1}^r z_l^{i_l t_l + j_l}\\
    =\ & \sum_{i_1}\cdots \sum_{i_r} \Abf\left[i_1, \cdots, i_r\right] \prod_{l=1}^r z_l^{i_l t_l} \\
    & \sum_{j_1}\cdots \sum_{j_r}\Bbf\left[j_1, \cdots, j_r\right] \prod_{l=1}^r z_l^{j_l}\\
    =\ & A(\zbf^{\tbf})B(\zbf).
    \eea
  \end{proof}

  The multiplication of two polynomials amounts to the convolution of two arrays, and the convolution may typically ruin the polyphase property. But Proposition \ref{prop:kron} can maintain the polyphase property of the array by utilizing the sparse polynomial $A(\zbf^{\tbf})$ to avoid the summation of the coefficients \textemdash this explains why the Kronecker product is necessary.
   
    \begin{proof}[Proof of Theorem \ref{thm:binary}] For notational simplicity, denote $a =A(\zbf^{\tbf})$, $b=B(\zbf^{\tbf})$,  $c^{\prime}=C(\zbf)$ and $d^{\prime}=D(\zbf)$, where $A(\zbf)$, $B(\zbf)$, $C(\zbf)$, and $D(\zbf)$ as the polynomials of $\Abf, \Bbf, \Cbf,$ and $\Dbf$, respectively. %and $*$ is as defined in (\ref{eqastr}). 
    By Definition \ref{def:GCA_3}, we have 
   \ben
   aa^*+bb^* = 2\prod_{i=1}^r s_i z_i^{t_i(s_i-1)}, \ 
   c^{\prime}c^{\prime *}+d^{\prime}d^{\prime *} = 2\prod_{i=1}^r t_i z_i^{t_i-1}.
   \een
   Then denote $c = \frac{1}{2} (c^{\prime}+d^{\prime})$, $d = \frac{1}{2} (c^{\prime}-d^{\prime})$, $e = E(\zbf)$, $f = F(\zbf)$. % Here $E(\zbf)$ and $F(\zbf)$ are polynomials of $\Ebf$ and $\Fbf$ as defined in , respectively. 
   By Proposition \ref{prop:kron}, it follows from (\ref{dymond})  that 
    \ben
    e = ac+bd, \quad f = b^*c-a^*d.
    \een
     Hence, by Lemma \ref{lem:2Identity} we have
  \bea
  &ee^*+ff^* \\
  =\ &(aa^*+bb^*)(cc^*+dd^*) \\
  =\ &(aa^*+bb^*)\left[\frac{1}{4}(c^{\prime}+d^{\prime})(c^{\prime}+d^{\prime})^{*}+\frac{1}{4}(c^{\prime}-d^{\prime})(c^{\prime}-d^{\prime})^{*}\right] \\
  =\ &\frac{1}{2} (aa^*+bb^*)(c^{\prime}c^{\prime *}+d^{\prime}d^{\prime *}) \\
  =\ &2\prod_{i=1}^{r} s_i t_i z_i^{s_it_i-1},
  \eea
  which satisfies the condition of autocorrelation complementarity by  (\ref{PCondition}). Since $\Cbf$ and $\Dbf$ are binary, $\frac{1}{2} \left(\Cbf+\Dbf\right)$ and $\frac{1}{2} \left(\Cbf-\Dbf\right)$ are disjoint, i.e., for any element of $\frac{1}{2} \left(\Cbf+\Dbf\right)$ being binary, the corresponding element of  $\frac{1}{2} \left(\Cbf-\Dbf\right)$ is $0$,  and vice versa. Therefore, the entries of $\Ebf$ and $\Fbf$ are also binary.
  \end{proof} 
  
  \begin{remark}
    The new proof of Theorem \ref{thm:binary} illustrates the benefits of introducing the perspective of commutative ring. Indeed, all the constructions in the paper are based on the four lemmas established from the perspective. The remainder is to maintain the polyphase property: the risk of ruining the polyphase property introduced by the composition $\cdot$ is hedged by the Kronecker product, while the risk introduced by the composition $+$ shall be avoided by introducing zeros in the proper positions.
  \end{remark}

\section{Constructions of Polyphase GCA Pairs} \label{SEC3}
% %   The binary GCM pairs constructed above can be combined with the Alamouti code \cite{alamouti1998simple} to achieve omnidirectional broadcasting of common message for a MIMO system with a uniform rectangular array (URA) \cite{li2021construction}. However, the sizes of antenna arrays are limited by the sparse existence pattern of the binary GCM pairs. 
%   One may use Theorem \ref{thm:binary} to construct a polyphase GCA pair$\{\Ebf, \Fbf\}$ even if $\{\Abf, \Bbf\}$ is polyphase and $\{\Cbf, \Dbf\}$ is binary. But if $\{\Cbf, \Dbf\}$ is no longer binary, the construction fails because the entries of $\{\Ebf, \Fbf\}$ may be illegal. 
  The seed sequences fed into Theorem \ref{thm:binary} are discovered by computational search \cite{golay1961complementary}. The lengths are $2$, $10$ and $26$, referred as basic binary Golay number (BBGN). Then the size in each dimension of the binary GCA pair constructed by Theorem \ref{thm:binary} is a binary Golay number, i.e., a number of form \ben \{ 2^{a}10^{b}26^{c}| a,b,c \in {\mathbb Z}^+ \}, \label{eqbinGolayLen}\een 
  where $a, b, c$ are non-negative integers.
  Seeking for a denser existence pattern, in our previous work \cite{li2021construction}, we constructed the quaternary GCM pairs based on two quaternary Golay sequence pairs.
  \begin{theorem}[\cite{li2021construction}] \label{thm:GCM_v1}
    Given $\{\abf, \bbf\}$ a pair of polyphase Golay sequences of length $L$, $\{\cbf, \dbf\}$ a pair of polyphase Golay sequences of length $M$, the two matrices
    \ben
    \Abf = \begin{bmatrix}
      \abf\cbf^T\\ \bbf\dbf^T
    \end{bmatrix}, \quad
    \Bbf = \begin{bmatrix}
      -\abf\dbf^{*T}\\ \bbf\cbf^{*T}
    \end{bmatrix}, 
    \een
    constitute a polyphase GCM pair of size $2L\times M$, where $[\cdot]^T$ denotes transposing a matrix.
  \end{theorem}
  
   The feasible sizes of the above quaternary GCM pairs are $2g_q^{(1)}\times g_q^{(2)}$ or $g_q^{(1)}\times 2g_q^{(2)}$, where $g_q^{(1)}$ and $g_q^{(2)}$, referred as quaternary Golay numbers, are of the form \cite{craigen2002complex}
   \begin{align} & g_q^{(1)},g_q^{(2)} \in  \label{eqquaGolayLen}\\
   &\left\{ 2^{a+u}3^{b}5^{c}11^{d}13^{e} \left| \begin{array}{r}a, b, c, d, e, u \in {\mathbb Z}^+, u \leq c+e, \\  b+c+d+e \leq a+2u+1 \end{array} \right. \right\}, \nonumber\end{align}  
   which is a set denser than (\ref{eqbinGolayLen}). Similar to the binary counterpart, the lengths of the quaternary seed sequences, i.e., $3$, $5$, $11$ and $13$, are referred as basic quaternary Golay number (BQGN).
  
% The proof  for Theorem \ref{thm:GCM_v1} provided in \cite{li2021construction} is based on the property of the rank-one matrix. 
We generalize Theorem \ref{thm:GCM_v1} to a multi-dimensional version using the concept of polynomial ring, which yields a denser feasible set than Theorem \ref{thm:GCM_v1}.

  \begin{theorem} \label{thm:N-ary}
    Given two polyphase GCA pairs $\{\Abf, \Bbf\}$ of size $s_1\times \cdots\times s_r$ and $\{\Cbf, \Dbf\}$ of size $t_1\times \cdots\times t_r$ respectively, a polyphase GCA pair $\{\Ebf, \Fbf\}$ of size $s_1t_1\times \cdots\times 2s_it_i\times \cdots \times s_rt_r$ can be constructed as
    \bea
    \Ebf &= \Abf \otimes \Cbf \vert \Bbf \otimes \Dbf,\\
    \Fbf &= \Bbf^* \otimes \Cbf \vert -\Abf^* \otimes \Dbf.
    \eea
    where $\vert$ denotes concatenating two arrays in the $i$-th dimension.
  \end{theorem}
  \begin{proof}
    % Let $\Ubf$ and $\Vbf$ be the arrays which are vectors $\left[1, 0\right]$ and $\left[0, 0\right]$ in the $i$-th dimension respectively and trivial in the other $r-1$ dimensions, i.e., the sizes in those dimensions are $1$. Then 
    % \ben
    % U(\zbf)U^*(\zbf) + V(\zbf)V^*(\zbf) = z_i.
    % \een
    % Thus $\{\Ubf, \Vbf\}$ is a GCA pair. 
    Suppose
    \ben \label{eq36}
     \Abf^{\prime} := \Abf \vert {\bf 0}, \quad \Bbf^{\prime} := {\bf 0} \vert \Bbf
    \een
    where $\bf{0}$ is a zero array of size $t_1\times \cdots \times t_r$. It's obvious that $\{\Abf^{\prime}, \Bbf^{\prime}\}$ is a disjoint GCA pair. Let $a$ and $b$ be $A^\prime(\zbf^{\tbf})$ and $B^\prime(\zbf^{\tbf})$ respectively, $c$ and $d$ be $C(\zbf)$ and $D(\zbf)$ respectively. Then
    \ben
    aa^*+bb^* = 2 z_i^{t_{i}s_{i}}\prod_{j=1}^{r} s_j z_j^{t_j(s_j-1)}, \ cc^*+dd^* = 2\prod_{j=1}^r t_j z_j^{t_j-1}.
    \een
    Noting that
    \ben
    \Ebf = \Abf^\prime \otimes \Cbf + \Bbf^\prime \otimes \Dbf, \quad \Fbf = \Bbf^{\prime *} \otimes \Cbf - \Abf^{\prime *} \otimes \Dbf,
    \een
    let $e$ and $f$ be $E(\zbf)$ and $F(\zbf)$ respectively, then by Lemma \ref{lem:2Identity},
    \ben
    ee^*+ff^* = (aa^*+bb^*)(cc^*+dd^*) = 4 z_i^{s_it_i}\prod_{j=1}^{r} s_j t_j z_j^{s_{j}t_{j}-1},
    \een
    which satisfies the GCA Definition \ref{def:GCA_3}. %condition of autocorrelation complementarity. 
    Besides, the entries remain polyphase since $\Abf^{\prime}$ and $\Bbf^{\prime}$ are disjoint.
  \end{proof}

By Theorem \ref{thm:N-ary} we can construct the quaternary GCM pairs of more feasible sizes than Theorem \ref{thm:GCM_v1}. For example, from quaternary Golay sequence pairs $\{\Abf, \Bbf\}$ of size $3\times 1$ and a trivial quaternary Golay sequence pair $\{\Cbf, \Dbf\}$ of size $1\times 1$, we first construct two quaternary GCM pairs of size $3\times 2$ by Theorem \ref{thm:N-ary}, based on which we can use Theorem \ref{thm:N-ary} again to construct a quaternary GCM pair of size $9\times 8$. Such a GCM, however, cannot be constructed by Theorem \ref{thm:GCM_v1}, since $9$ is not a quaternary Golay number [cf. (\ref{eqquaGolayLen})]. % The quaternary Golay sequence pairs of length $2$, $3$, $5$, $11$, $13$ discovered by computational search \cite{craigen2002complex}. 
  
  Theorem \ref{thm:N-ary} can be viewed as a polyphase and multi-dimensional version of the concatenation method in \cite[General Properties 11)]{golay1961complementary}. An interleaved alternately version is also straightforward to prove by modification of the polynomials, but it does not produce more feasible sizes.

% \subsection{Construction of Disjoint GCA Sets}
  One major technique of the construction in Theorem \ref{thm:N-ary} is to construct a disjoint GCA set [cf. (\ref{eq36})] for obviating the summation of unit roots. And there is more flexible approaches. In \cite{craigen2002complex} it was proved that $g_b g_q^{(1)} g_q^{(2)}$ is a quaternary Golay number if $g_b$ is a binary Golay number and $g_q^{(1)}, g_q^{(2)}$ are quaternary Golay numbers, of which the main idea is to find a weight-deficient Golay sequence pair with special structure by exploiting the well-known symmetry of the binary Golay sequence pair \cite{golay1961complementary}. Inspired by this idea, first we prove the symmetry of a binary GCA pair in Proposition \ref{prop:sym}, then utilize it to generalize the method in \cite{craigen2002complex} to construct the polyphase GCA pairs in Theorem \ref{thm:G-N-ary}. %

  Here are some notations to be used:
  \bea
  \Abf\left[\ibf\right] &:= \Abf\left[i_1, \cdots, i_r \right], \\
  \Abf\left[\ibf+\jbf\right] &:= \Abf\left[i_1+j_1, \cdots, i_r+j_r \right], \\
  \ibf+1 &:= \ibf^{\prime}\ where \ \sum_{k=1}^{r}i^{\prime}_{k}\prod_{l=1}^{k-1}s_{l} = 1+\sum_{k=1}^{r}i_{k}\prod_{l=1}^{k-1}s_{l}\\
  \sum_{\jbf=\bf{0}}^{\ibf} &:= \sum_{j_1=0}^{i_1} \cdots \sum_{j_r=0}^{i_r}, \quad
  \prod_{\jbf=\bf{0}}^{\ibf} := \prod_{j_1=0}^{i_1} \cdots \prod_{j_r=0}^{i_r}, \\
  N_{\ibf} &:= \prod_{k=1}^{r} \left(i_k+1\right), \quad \zbf^{\ibf} := z_1^{i_1} \cdots z_r^{i_r}.\\
  \ibf < \jbf &:= \sum_{k=1}^{r}i_{k}\prod_{l=1}^{k-1}s_{l} < \sum_{k=1}^{r}j_{k}\prod_{l=1}^{k-1}s_{l}.
  \eea

  \begin{Proposition} \label{prop:sym}
    Given a nontrivial binary GCA pair $\{\Abf, \Bbf\}$ of size $s_1 \times \cdots \times s_r$, for any given index $\ibf$ satisfying $\bf{0} \leq \ibf \leq \sbf-{\bf{1}}$, 
    \ben \label{symmetry}
    \Abf\left[\sbf-\bf{1}-\ibf\right] \Abf\left[\ibf\right] \Bbf\left[\sbf-\bf{1}-\ibf\right] \Bbf\left[\ibf\right] = -1.
    \een
  \end{Proposition}
  \begin{proof}
    For any $\bf{0} \leq \ibf < \sbf-{\bf{1}}$, it follows from (\ref{corr}) and (\ref{ACondition}) that
    \ben \label{coff}
    \sum_{\jbf=\bf{0}}^{\ibf} \Abf\left[\sbf-\bf{1}-\ibf+\jbf\right] \Abf\left[\jbf\right] + \Bbf\left[\sbf-\bf{1}-\ibf+\jbf\right] \Bbf\left[\jbf\right] = 0.
    \een
    The left side of (\ref{coff}) must be a summation of $N_{\ibf}$ "$1$"s and $N_{\ibf}$ "$-1$"s. Therefore,
    \bea \label{power}
    & \prod_{\jbf=\bf{0}}^{\ibf} \Abf\left[\sbf-\bf{1}-\ibf+\jbf\right] \Abf\left[\jbf\right] \Bbf\left[\sbf-\bf{1}-\ibf+\jbf\right] \Bbf\left[\jbf\right] \\
    =& \prod_{\jbf=\bf{0}}^{\ibf} \Abf\left[\sbf-\bf{1}-\jbf\right] \Abf\left[\jbf\right] \Bbf\left[\sbf-\bf{1}-\jbf\right] \Bbf\left[\jbf\right] \\
    =& (-1)^{N_{\ibf}}.
    \eea
    By (\ref{power}), (\ref{symmetry}) holds when $\ibf = \left[0, 0, \cdots, 0\right]$. Suppose (\ref{symmetry}) holds for $\forall \ \ibf \leq \kbf$, then by (\ref{power}), for $\ibf = \kbf+1$, 
    \ben \label{eq:induction_k}
    \prod_{\jbf=\bf{0}}^{\kbf+1} \Abf\left[\sbf-\bf{1}-\jbf\right] \Abf\left[\jbf\right] \Bbf\left[\sbf-\bf{1}-\jbf\right] \Bbf\left[\jbf\right]
    = (-1)^{N_{\kbf+1}}.
    \een
    The indices $\jbf$ in (\ref{eq:induction_k}) satisfy $\jbf \leq \kbf$ except for $\jbf = \kbf+1$. Hence
    \bea
    &\Abf\left[\sbf-{\bf{1}}-\kbf-1\right] \Abf\left[\kbf+1\right] \Bbf\left[\sbf-{\bf{1}}-\kbf-1\right] \Bbf\left[\kbf+1\right]\\
    =& \ \frac{(-1)^{N_{\kbf+1}}}{(-1)^{N_{\kbf+1}-1}}\\
    =& \ -1.
    \eea
    % By (\ref{power}), when $\ibf = \left[0, 0, \cdots, 0\right]$, (\ref{symmetry}) holds. Next perform the following induction steps:
    % \begin{enumerate}[Step 1:]
    %   \item Set $d=1$, $\Inum=\{\left[0, 0, \cdots, 0\right]\}$;
    %   \item Starting from each index in $\Inum$, apply induction in the $d$-th dimension and adding feasible indices into $\Inum$ for which (\ref{coff}) holds. Set $d=d+1$;
    %   \item If $d>r$, break; else, go to Step 2.
    % \end{enumerate}
    
    Thus (\ref{symmetry}) holds for any ${\bf{0}} \leq {\ibf} < \sbf-{\bf{1}}$.
    Noting that (\ref{symmetry}) is identical when $\ibf = \bf{0}$ or $\ibf={\sbf}-\bf{1}$, thus (\ref{symmetry}) holds for any ${\bf{0}} \leq {\ibf} \leq \sbf-{\bf{1}}$.
  \end{proof}
  Based on Lemma \ref{lem:2Identity} and Proposition \ref{prop:sym}, we have the following method to assign zeros in the GCA pairs.

  \begin{Proposition} \label{prop:disjoint}
    Suppose $\{\Abf, \Bbf\}$ a binary GCA pair and $\{\Pbf, \Qbf\}$ a GCA pair with entries $\{1, -1, 0\}$, and for any given position $\ibf$, exactly one of $\Pbf[\ibf]$, $\Qbf[\ibf]$, $\Pbf^*[\ibf]$ and $\Qbf^*[\ibf]$ equals $\pm 1$ while the others equal $0$. Then there exists such $\{\Pbf, \Qbf\}$ of size $s_1\times \cdots \times s_r$ if and only if there exists such $\{\Abf, \Bbf\}$ of the same size.
  \end{Proposition}
  \begin{proof}
  The "if" part: set
  \bea 
    \Pbf &= \frac{1}{4} {\left[\Abf+\Bbf+\left(\Bbf^*-\Abf^*\right)\right]},\\
    \Qbf &= \frac{1}{4} {\left[\Abf+\Bbf-\left(\Bbf^*-\Abf^*\right)\right]}.
  \eea
    Suppose a trivial GCA pair $\Cbf = \Dbf = 1$. Set
    \ben
    \Ebf = \Abf \otimes \Cbf + \Bbf \otimes \Dbf, \quad \Fbf = \Bbf^* \otimes \Cbf - \Abf^* \otimes \Dbf,
    \een
    By Lemma \ref{lem:2Identity}, set $a = A(\zbf)$, $b = B(\zbf)$, $c = C(\zbf)$, $d = D(\zbf), e = E(\zbf)$, $f = F(\zbf)$, then we have
    \ben
    E(\zbf)E^*(\zbf) + F(\zbf)F^*(\zbf) = 4\prod_{i=1}^r s_i z_{i}^{s_i-1}.
    \een
    Noting that 
    \ben
    \Pbf = \frac{1}{4}(\Cbf \otimes \Ebf + \Dbf \otimes \Fbf), \quad \Qbf = \frac{1}{4}(\Dbf^* \otimes \Ebf - \Cbf^* \otimes \Fbf).
    \een
    Using Lemma \ref{lem:2Identity} again, set $a = C(\zbf)$, $b = D(\zbf)$, $c = \frac{1}{4}E(\zbf)$, $d = \frac{1}{4}F(\zbf)$, $e = P(\zbf)$, $f = Q(\zbf)$, then we have
    \ben
    ee^*+ff^* = (aa^*+bb^*)(cc^*+dd^*) = \frac{1}{2}\prod_{i=1}^r s_i z_{i}^{s_i-1}.
    \een
    By Proposition \ref{prop:sym}, without loss of generality, suppose $\Abf[\ibf] = -1$ and $\Bbf[\ibf]=\Abf^*[\ibf]=\Bbf^*[\ibf]=1$ for a given index $\ibf$, then $\Pbf[\ibf]=\Qbf[\ibf]=\Pbf^*[\ibf]=0$ and $\Pbf^*[\ibf]=1$. Thus $\{\Pbf, \Qbf\}$ is a weight-deficient GCA pair satisfying the required structure.
    
    The "only if" part: Set
    \ben 
    \Xbf = \Pbf+\Qbf,\quad
    \Ybf = \Qbf^{*} - \Pbf^{*}, 
    \een
    then $\{\Xbf, \Ybf\}$ is autocorrelation complementary and by the structure of $\{\Pbf, \Qbf\}$, $\{\Xbf, \Ybf\}$ is disjoint. Thus $\{\Xbf, \Ybf\}$ is a disjoint GCA pair. Next we set
    \ben
    \Abf = \Xbf + \Ybf, \quad
    \Bbf = \Xbf - \Ybf,
    \een
    then $\{\Abf, \Bbf\}$ is a binary GCA pair.
  \end{proof}
  
  Based on the structure of the weight-deficient GCA pairs $\{\Pbf, \Qbf\}$, we propose a construction of the polyphase GCA pairs, which is a generalization of Theorem \ref{thm:N-ary} and the construction of the quaternary Golay sequence pairs in \cite{craigen2002complex}.

  \begin{theorem} \label{thm:G-N-ary}
    Given a nontrivial binary GCA pair $\{\Abf, \Bbf\}$ of size $s_1\times \cdots \times s_r$, two polyphase GCA pairs $\{\Cbf, \Dbf\}$ and $\{\Ebf, \Fbf\}$ of size $t_1\times \cdots \times t_r$ and $u_1\times \cdots \times u_r$ respectively, suppose
    \bea 
    \Pbf &= \frac{1}{4} {\left[\Abf+\Bbf+\left(\Bbf^*-\Abf^*\right)\right]},\\
    \Qbf &= \frac{1}{4} {\left[\Abf+\Bbf-\left(\Bbf^*-\Abf^*\right)\right]},\\
    \Xbf &= \Pbf \otimes \Cbf + \Qbf \otimes \Dbf,\\
    \Ybf &= \Qbf^* \otimes \Cbf - \Pbf^* \otimes \Dbf,\\
    \Gbf &= \Xbf \otimes \Ebf + \Ybf \otimes \Fbf,\\
    \Hbf &= \Ybf^* \otimes \Ebf - \Xbf^* \otimes \Fbf,
    \eea
    then $\{\Gbf, \Hbf\}$ is a polyphase GCA pair of size $s_1t_1u_1\times \cdots \times s_rt_ru_r$.
  \end{theorem}
  \begin{proof}
    As done in the proof of the previous theorems, the autocorrelation complementarity follows from Lemma \ref{lem:2Identity} recursively. And due to the structure of $\Pbf, \Qbf$ demonstrated in Proposition \ref{prop:disjoint}, $\Xbf$ and $\Ybf$ are disjoint, then $\Gbf$ and $\Hbf$ are still polyphase. Hence $\{\Gbf, \Hbf\}$ is a polyphase GCA pair.
  \end{proof}

  If in Theorem \ref{thm:G-N-ary} $\Abf$ and $\Bbf$ are vectors $\left[1, 1\right]$ and $\left[1, -1\right]$ in the $i$-th dimension respectively and trivial in the other $r-1$ dimensions, Theorem \ref{thm:G-N-ary} degenerates into Theorem \ref{thm:N-ary}; thus, Theorem \ref{thm:G-N-ary} is a generalization of Theorem \ref{thm:N-ary}.
  
   Theorem \ref{thm:G-N-ary} can be used to construct quaternary GCM pairs of more feasible sizes than Theorem \ref{thm:N-ary}. For example, from a binary Golay sequence pair of size $1\times 10$ and two quaternary Golay sequence pairs of size $3\times 1$, a matrix of size $9\times 10$ can be constructed by Theorem \ref{thm:G-N-ary}, which is infeasible in Theorem \ref{thm:N-ary} since $9$ is not a quaternary Golay number and there exists not a quaternary GCM pair of size $3\times 5$ (the nonexistence of quaternary GCM pair of size $3\times 5$ follows from the nonexistence of quaternary Golay sequence pair of length $15$ \cite{bright2021complex} since one may reshape a GCM pair to a Golay sequence pair \cite{jedwab2007golay}).

 % \begin{remark}
    In Theorem \ref{thm:G-N-ary}, the binary GCA pair is like a binder to glue two polyphase GCA pairs together. From this viewpoint, we give the feasible sizes of the quaternary GCA pairs in the following corollary .
%  \end{remark}

%Finally, we summarize the feasible sizes of the quaternary GCA pairs in Corollary \ref{Qsize}.
  \begin{Corollary} \label{Qsize}
    There exist quaternary GCA pairs of size $s_1\times \cdots\times s_r$ where $s_1\cdots s_r = 2^{a+u}3^{b}5^{c}11^{d}13^{e}$, $a, b, c, d, e, u \geq 0, b+c+d+e \leq a+2u+1, u \leq c+e$, and each of $u$ factors $10$ or $26$ shouldn't be factorized into different dimensions.
  \end{Corollary}
  \begin{proof}
    First, we consider feeding into Theorem \ref{thm:G-N-ary} the seed Golay sequences, whose lengths are BBGN or BQGN, i.e., $\{2, 10, 26\}$ or $\{3, 5, 11, 13\}$. Noting that $10 = 2\times 5$ and $26 = 2\times 13$, thus $\prod_{i=1}^{r}s_i$ is of form $2^{a+u}3^{b}5^{c}11^{d}13^{e}$ where $u\leq c+e$ is derived from the decomposition of $10$ and $26$. Viewing the role of BBGNs as the binder to glue BQGNs, we need at least $n-1$ BBGNs to combine $n$ BQGNs. In the decomposition of $2^{a+u}3^{b}5^{c}11^{d}13^{e}$ there exist $a+u$ BBGNs and $b+c+d+e-u$ BQGNs, thus $a+u \geq b+c+d+e-u-1 \Rightarrow b+c+d+e \leq a+2u+1$. Second, due to the nonexistence of the binary GCM pair of size $2\times 5$ or $2\times 13$, which has been verified by exhaustive computational search in \cite{jedwab2007golay}, if any of $u$ factors $10$ or $26$ is decomposed into different dimensions, the number of BBGNs decreases.
  \end{proof}

An example of the last constraint: we can construct a quaternary GCM pair of size $9\times 10$ by Theorem \ref{thm:G-N-ary}, but can not construct a quaternary GCM pair of size $18\times 5$.

\section{Constructions of Polyphase GCA Quads} \label{SEC4}

Besides the GCM pairs, the GCM quads can be combined with the $4\times 4$ space-time block code (STBC) \cite{Tarokh1999Space} for omnidirectional transmission \cite{li2021construction}. Previously we proposed a construction of polyphase GCM sets in \cite{jiang2019autocorrelation} from two polyphase Golay sequence sets.
\begin{theorem}[\cite{jiang2019autocorrelation}] \label{thm:naive_cross}
  For $\{{\abf}_1, \cdots, \abf_{L_1}\}$ a polyphase Golay sequence set of length $M$ and $\{{\bbf}_1, \cdots, \bbf_{L_2}\}$ a polyphase Golay sequence set of length $N$, set
  \ben
  \Cbf_{ij} = \abf_i \bbf_j^T, \quad 1\leq i\leq L_1, 1\leq j\leq L_2,
  \een
  then $\{\Cbf_{ij} \mid 1\leq i\leq L_1, 1\leq j\leq L_2\}$ is a polyphase GCM set of size $M\times N$.
\end{theorem}
Let $L_1=L_2=2$, then the polyphase GCM quads of size $M \times N$ can be constructed, e.g., a quaternary GCM quad of size $3\times 3$ can be constructed from two quaternary Golay sequence pairs of length $3$. 

An interpretation of Theorem \ref{thm:naive_cross} based on the property of the rank-one matrix was given in \cite{jiang2019autocorrelation}. From the commutative ring perspective, in Section \ref{SEC2} we propose the basic identity in Lemma \ref{lem:mnIdentity} behind the construction, and generalize Theorem \ref{thm:naive_cross} slightly to construct the polyphase GCA set. 
\begin{theorem} \label{thm:cross}
  Given two polyphase GCA sets $\{\Abf_1, \cdots, \Abf_m\}$ and $\{\Bbf_1, \cdots, \Bbf_n\}$ of size $s_1\times \cdots\times s_r$ and $t_1\times \cdots\times t_r$ respectively, set
  \ben
  \Cbf_{ij} = \Abf_i \otimes \Bbf_j, \quad 1\leq i\leq m, 1\leq j\leq n,
  \een
  then $\{\Cbf_{ij} \vert 1\leq i\leq m, 1\leq j\leq n\}$ is a polyphase GCA set of size $s_1t_1 \times s_2t_2 \times \cdots \times s_rt_r$.
\end{theorem}

\begin{proof}
  By Lemma \ref{lem:mnIdentity}, let $a_i = A_{i}(\zbf^{\tbf})$, $b_j = B_{j}(\zbf)$, $c_{ij} = C_{ij}(\zbf)$, and the involutive automorphism $*$ is defined as in (\ref{eqastr}), then 
  \bea
  \sum_{i, j} c_{ij}c_{ij}^* =\ &\sum_{i} a_ia_i^* \sum_{j} b_jb_j^*\\
  =\ &m \prod_{k=1}^{r}s_k z_k^{t_k(s_k-1)} \cdot n \prod_{k=1}^{r}t_k z_k^{t_k-1}\\
  =\ &mn \prod_{k=1}^{r} s_kt_k z_k^{s_kt_k-1}.
  \eea
  Thus $\{\Cbf_{ij} \vert 1\leq i\leq m, 1\leq j\leq n\}$ is a polyphase GCA set.
\end{proof}
By Theorem \ref{thm:cross}, we can construct the quaternary GCM quads of more feasible sizes than Theorem \ref{thm:naive_cross}, e.g, from two quaternary GCM pairs of size $3\times 6$, we can construct a quaternary GCM quad of size $9\times 36$, which can not be constructed by Theorem \ref{thm:naive_cross} since $9$ is not a quaternary Golay number.

We note that to keep the entries polyphase, Theorem \ref{thm:G-N-ary} and Theorem \ref{thm:cross} use different methods. Theorem \ref{thm:G-N-ary} assigns zeros properly and uses Lemma \ref{lem:2Identity} recursively to obviate the summation of unit roots, at the expense of multiplying the size of arrays in one dimension with a binary Golay number. The construction in Theorem \ref{thm:cross} is polyphase intrinsically for no composition $+$ in (\ref{eq:cross}) of Lemma \ref{lem:mnIdentity}, though at the expense of enlarging the cardinality of the array set. As a compromise between the cardinality of the set and the size of the arrays, we propose Lemma \ref{lem:compromise} in Section \ref{SEC2}.

Based on Lemma \ref{lem:4Identity} and Lemma \ref{lem:compromise}, we construct the polyphase GCA quads of denser existence pattern than Theorem \ref{thm:cross}.

To use Lemma \ref{lem:4Identity} to construct the polyphase GCA quads, the summation of unit roots should be forbidden by interleaving alternately or concatenating zeros as mentioned in \cite{yang1989composition}. The method of assigning zeros are explained in Proposition \ref{prop:interleave} and \ref{prop:concate_zero}.
\begin{Proposition} \label{prop:interleave}
  Given a polyphase GCA quad $\{\Abf, \Bbf, \Cbf, \Dbf\}$ where the size of $\Abf$ and $\Bbf$ is $s_1\times \cdots\times (s_i+1)\times \cdots\times s_r$ and the size of $\Cbf$ and $\Dbf$ is $s_1\times \cdots\times s_i \times \cdots\times s_r$, suppose
  \bea
  &\Ebf = \Abf/{\bf{0}}_{s_i}, \qquad \Gbf = \Bbf/{\bf{0}}_{s_i}, \\ &\Fbf = {\bf{0}}_{s_i+1}/\Cbf, \quad \Hbf = {\bf{0}}_{s_i+1}/\Dbf,
  \eea
  where $/$ denotes interleaving alternately two arrays in the $i$-th dimension, ${\bf{0}}_{s_i}$ $\left({\bf{0}}_{s_i+1}\right)$ denotes the zero array of the same size as $\Cbf$ $\left(\Abf\right)$.
  Then $\{\Ebf, \Fbf, \Gbf, \Hbf\}$ is a weight-deficient GCA quad of size $s_1\times \cdots\times \left(2s_i+1\right)\times \cdots\times s_r$.
\end{Proposition}
\begin{proof}
  Note that
  \bea
  e := E(\zbf)&=A(z_1, \cdots, z_i^2, \cdots, z_r) := a,\\
  g := G(\zbf)&=B(z_1, \cdots, z_i^2, \cdots, z_r) := b,\\
  f := F(\zbf)&=z_{i}C(z_1, \cdots, z_i^2, \cdots, z_r) := z_{i}c,\\
  h := H(\zbf)&=z_{i}D(z_1, \cdots, z_i^2, \cdots, z_r) := z_{i}d.
  \eea
  Define a polynomial ring including the polynomials of both positive and negative powers, and define its involutive automorphism $\star$ as a map from $A(\zbf)$ to $\overline{A(\zbf^{-1})}$. [Do not confuse it with the involutive automorphism $*$ defined by (\ref{eqastr}).]
  Then
  \bea
  ee^\star+ff^\star +gg^\star+hh^\star &\overset{(a)}{=} aa^\star+bb^\star+cc^\star+dd^\star \\
  &= 2(2s_i+1)\prod_{j=1, j\neq i}^r s_j.
  \eea
  where $\overset{(a)}{=}$ holds because $z_i z_i^\star=1$.
  By Definition \ref{def:GCA_2}, we have completed the proof.
\end{proof}

If the array quad $\{\Abf, \Bbf, \Cbf, \Dbf\}$ in Lemma \ref{prop:interleave} is trivial in all but the $i$-th dimensions, then it is named base sequences \cite{dokovic1998aperiodic}, denoted as $BS(m+1, m)$ where $m=s_i$. In \cite{dokovic1998aperiodic}, it was conjectured that there exists binary $BS(m+1, m)$ for any integer $m\geq1$, which has been verified for $m\leq38$ by computational search \cite{djokovic2010base}. Besides, there exists $BS(g+1, g)$ where $g$ is a polyphase Golay number \cite{turyn1974hadamard}.

\begin{Proposition} \label{prop:concate_zero}
  Given a GCA quad $\{\Abf, \Bbf, \Cbf, \Dbf\}$ where the size of $\Abf$ and $\Bbf$ is $s_1\times \cdots\times s_j\times \cdots\times s_r$ and the size of $\Cbf$ and $\Dbf$ is $s_1\times \cdots\times s_j^{\prime} \times \cdots\times s_r$, suppose
  \bea
  &\Ebf = \Abf \vert {\bf{0}}_{s_{j}^\prime} \vert {\bf{0}}_{s_{j}^\prime} \vert \Bbf \quad
  &\Gbf = \Abf \vert {\bf{0}}_{s_{j}^\prime} \vert {\bf{0}}_{s_{j}^\prime} \vert {-\Bbf}\\
  &\Fbf = {\bf{0}}_{s_j} \vert \Cbf \vert \Dbf \vert {\bf{0}}_{s_j} \quad
  &\Hbf = {\bf{0}}_{s_j} \vert \Cbf \vert {-\Dbf} \vert {\bf{0}}_{s_j}
  \eea
  where $\vert$ denotes concatenating two arrays in the $j$-th dimension, ${\bf{0}}_{s_j}$ $\left({\bf{0}}_{s_j^{\prime}}\right)$ denotes the zero array of the same size as $\Abf$ $\left(\Cbf\right)$. Then $\{\Ebf, \Fbf, \Gbf, \Hbf\}$ is a weight-deficient GCA quad of size $s_1\times \cdots\times 2\left(s_j+s_j^{\prime}\right)\times \cdots\times s_r$. 
\end{Proposition}
\begin{proof}
  Using the same definitions in the proof of Proposition \ref{prop:interleave}.
  Noting that
  \bea
  e := E(\zbf) &= z_j^{2s_j^\prime+s_j}A(\zbf)+B(\zbf) := xa+b,\\
  g := G(\zbf) &= z_j^{2s_j^\prime+s_j}A(\zbf)-B(\zbf) := xa-b,\\
  f := F(\zbf) &= z_j^{s_j^\prime+s_j}C(\zbf)+z_j^{s_j}D(\zbf) := y_{1}c + y_{2}d,\\
  h := H(\zbf) &= z_j^{s_j^\prime+s_j}C(\zbf)-z_j^{s_j}D(\zbf) := y_{1}c - y_{2}d,
  \eea
  where $xx^\star = y_{1}y_{1}^{\star} = y_{2}y_{2}^{\star} = 1$, then
  \bea
  ee^\star+ff^\star+gg^\star+hh^\star &= 2(aa^\star+bb^\star+cc^\star+dd^\star)\\
  &= 4(s_j+s_j^{\prime})\prod_{i=1, i\neq j}^r s_i.
  \eea
By Definition \ref{def:GCA_2}, we have completed the proof.
\end{proof}

Practically the GCA quad $\{\Abf, \Bbf, \Cbf, \Dbf\}$ in Proposition \ref{prop:concate_zero} can be composed of two GCA pair $\{\Abf, \Bbf\}$ and $\{\Cbf, \Dbf\}$.

In Proposition \ref{prop:interleave} and Proposition \ref{prop:concate_zero}, $\Ebf$, $\Fbf$, $\Gbf$ and $\Hbf$ are quasi-symmetric, i.e., the zeros of the flipped arrays occur in the same positions as the original arrays; $\Ebf$ and $\Gbf$, $\Fbf$ and $\Hbf$ are conjoint, i.e., zeros occur in the same positions; $\Ebf$ and $\Fbf$, $\Gbf$ and $\Hbf$ are disjoint. The above structure is sufficient to avoid the summation of the unit roots when they are combined in the following theorem.

\begin{theorem} \label{thm:quad}
  Construct  two weight-deficient GCA quads of size $m_1\times \cdots\times m_r$ and $n_1\times \cdots\times n_r$ according to Proposition \ref{prop:interleave} or Proposition \ref{prop:concate_zero}, which are denoted by $\{\Abf, \Bbf, \Cbf, \Dbf\}$ and $\{\Ebf, \Fbf, \Gbf, \Hbf\}$,  respectively. Suppose
  \bea
  \Pbf &= \Abf \otimes \Fbf^* - \Bbf^* \otimes \Ebf + \Cbf \otimes \Gbf + \Dbf \otimes \Hbf\\
  \Qbf &= \Abf^* \otimes \Ebf + \Bbf \otimes \Fbf^* - \Cbf \otimes \Hbf^* + \Dbf \otimes \Gbf^*\\
  \Rbf &= \Cbf^* \otimes \Ebf - \Dbf \otimes \Fbf + \Abf \otimes \Hbf^* + \Bbf \otimes \Gbf\\
  \Sbf &= -\Cbf \otimes \Fbf - \Dbf^* \otimes \Ebf + \Abf \otimes \Gbf^* - \Bbf \otimes \Hbf,
  \eea
  then $\{\Pbf, \Qbf, \Rbf, \Sbf\}$ is a polyphase GCA quad of size $m_1n_1\times \cdots\times m_{r}n_{r}$.
\end{theorem}
\begin{proof}
  By Lemma \ref{lem:4Identity}, let $a=A(\zbf^{\nbf})$, $b=B(\zbf^{\nbf})$, $c=C(\zbf^{\nbf})$, $d=D(\zbf^{\nbf})$, $e=E(\zbf)$, $f=F(\zbf)$, $g=G(\zbf)$, $h=H(\zbf)$, $p=P(\zbf)$, $q=Q(\zbf)$, $r=R(\zbf)$, $s=S(\zbf)$, and $*$ is as defined in (\ref{eqastr}), then
  \bea
  &pp^*+qq^*+rr^*+ss^*\\
  =\ &(aa^*+bb^*+cc^*+dd^*)(ee^*+ff^*+gg^*+hh^*)\\
  =\ &2\prod_{i=1}^r m_i z_i^{n_i(m_i-1)} \cdot 2\prod_{i=1}^r n_i z_i^{n_i-1}\\
  =\ &4\prod_{i=1}^r m_in_i z_i^{m_in_i-1}.
  \eea
  By Definition \ref{def:GCA_3}, $\{\Pbf, \Qbf, \Rbf, \Sbf\}$ is a GCA quad. Their entries are polyphase due to the quasi-symmetric, conjoint, and disjoint structure of the input arrays.
\end{proof}
\begin{remark}
  Feeding a quad of binary base sequences into Proposition \ref{prop:interleave} or a binary Golay sequence quad (consisting of two binary Golay sequence pairs) into Proposition \ref{prop:concate_zero}, we obtain a two weight-deficient Golay sequence quads and feed them into Theorem \ref{thm:quad}, then a binary GCM quad of size $m\times n$ or $1\times mn$ can be constructed, where $m, n\in \{2s+1\vert 0\leq s\leq 38, or \ s=g_b^{(1)}\} \cup \{2(g_b^{(2)}+g_b^{(3)})\}$ and $g_b^{(1)}, \cdots, g_b^{(3)}$ are binary Golay numbers.
\end{remark}
The sizes of the GCA quads constructed by Theorem \ref{thm:quad} may be constrained by the known existence pattern of the base sequences. We improve the pattern in the following theorem, which is the first application of Lemma \ref{lem:compromise}.
\begin{theorem} \label{thm:baseArray}
  Given a polyphase GCA quad $\{\Pbf, \Qbf, \Rbf, \Sbf\}$ of size $s_1\times \cdots\times s_r$ and a disjoint GCA pair $\{\Ibf, \Jbf\}$ of size $t_1\times \cdots\times t_r$. Suppose
  \bea
  \Pbf^{\prime} &= \Pbf \otimes \Ibf + \Qbf \otimes \Jbf, \quad
  \Qbf^{\prime} = \Pbf \otimes \Jbf^* - \Qbf \otimes \Ibf^*,\\
  \Rbf^{\prime} &= \Rbf \otimes \Ibf + \Sbf \otimes \Jbf, \quad
  \Sbf^{\prime} = \Rbf \otimes \Jbf^* - \Sbf \otimes \Ibf^*.
  \eea
  Then $\{\Pbf^{\prime}, \Qbf^{\prime}, \Rbf^{\prime},  \Sbf^{\prime}\}$ is a polyphase GCA quad of size $s_1t_1\times s_2t_2\times \cdots\times s_rt_r$
\end{theorem}
\begin{proof}
By Lemma \ref{lem:compromise}, set $m=4, n=2$, then we have
\ben
\sum_{i=1}^{2}\left(c_{i1}c_{i1}^{*} + d_{i1}d_{i1}^{*}\right) = \left(\sum_{i=1}^{4} a_i a_i^*\right) \left(\sum_{j=1}^{2} b_j b_j^*\right)
\een
if given 
\bea 
  c_{i1} &= a_{2i-1} b_{1} + a_{2i}b_{2},\\
  d_{i1} &= a_{2i-1} b_{2}^* - a_{2i}b_{1}^*.
 \eea
 Next let $a_1 = P(\zbf^{\tbf})$, $a_2 = Q(\zbf^{\tbf})$, $a_3 = R(\zbf^{\tbf})$, $a_4 = S(\zbf^{\tbf})$, $b_1 = I(\zbf)$, $b_2 = J(\zbf)$, $c_{11} = P^{\prime}(\zbf)$, $d_{11} = Q^{\prime}(\zbf)$, $c_{21} = R^{\prime}(\zbf)$, $d_{21} = S^{\prime}(\zbf)$, and $*$ is as defined in (\ref{eqastr}). Following the same procedures of algebra manipulations in the proof of prior theorem, we have 
 \ben
\sum_{i=1}^{2}\left(c_{i1}c_{i1}^{*} + d_{i1}d_{i1}^{*}\right) = 4\prod_{i=1}^r s_it_i z_i^{s_it_i-1}.
\een
By Definition \ref{def:GCA_3}, $\{\Pbf^{\prime}, \Qbf^{\prime}, \Rbf^{\prime},  \Sbf^{\prime}\}$ is a GCA quad. Their entries are polyphase due to the disjoint structure of $\{\Ibf, \Jbf\}$.
\end{proof}

The disjoint GCA pair in the above theorem can be the intermediate product $\{\frac{1}{2} \left(\Cbf+\Dbf\right), \frac{1}{2} \left(\Cbf-\Dbf\right)\}$ in Theorem \ref{thm:binary} or $\{\Xbf, \Ybf\}$ in Theorem \ref{thm:G-N-ary}.

Feeding into Theorem \ref{thm:baseArray} the binary GCM quad constructed in the remark following Theorem \ref{thm:quad} and a disjoint GCM pair of size $g_b^{(1)}\times g_b^{(2)}$, we have the following corollary about the existence pattern of binary GCM quads:
\begin{Corollary} \label{binary_GCM_quad_size}
  The feasible sizes of binary GCM quads may be of form $g_b^{(1)}m\times g_b^{(2)}n$ or $g_b^{(1)}\times g_b^{(2)}mn$ where $m, n\in \{2s+1\vert 0\leq s\leq 38, or \ s=g_b^{(3)}\} \cup \{2(g_b^{(4)}+g_b^{(5)})\}$ and $g_b^{(1)}, \cdots, g_b^{(5)}$ are binary Golay numbers. Specifically, the sizes within $78\times 78$ can be covered.
\end{Corollary}

% \begin{Conjecture}
%   There exist binary GCM quads of arbitrary sizes.
% \end{Conjecture}

By Theorem \ref{thm:quad}, the existence pattern of quaternary GCM quads is much denser than the binary counterpart. Specifically, we have the following corollary:
\begin{Corollary} \label{double_sum_size}
  There exist quaternary GCM quads of size $2s_1(t_2+t_3) \times 2t_1(s_2+s_3)$ where $s_1 \times s_2$, $s_1 \times s_3$, $t_2 \times t_1$ and $t_3 \times t_1$ are the feasible sizes of quaternary GCM pairs.
\end{Corollary}
\begin{proof}
\begin{enumerate}
  \item Construct four quaternary GCM pairs $\{\Abf_1, \Bbf_1\}$ of size $s_1 \times s_2$, $\{\Cbf_1, \Dbf_1\}$ of size $s_1 \times s_3$, $\{\Ebf_1, \Fbf_1\}$ of size $t_2 \times t_1$ and $\{\Gbf_1, \Hbf_1\}$ of size $t_3 \times t_1$ by Theorem \ref{thm:G-N-ary}.
  \item By Proposition \ref{prop:concate_zero}, use $\{\Abf_1, \Bbf_1, \Cbf_1, \Dbf_1\}$ and $\{\Ebf_1, \Fbf_1, \Gbf_1, \Hbf_1\}$ to construct two quasi-symmetric, conjoint and disjoint GCA quads $\{\Abf_2, \Bbf_2, \Cbf_2, \Dbf_2\}$ of size $s_1 \times 2(s_2+s_3)$ and $\{\Ebf_2, \Fbf_2, \Gbf_2, \Hbf_2\}$ of size $2(t_2+t_3) \times t_1$ respectively.
  \item Feed $\{\Abf_2, \Bbf_2, \Cbf_2, \Dbf_2\}$ and $\{\Ebf_2, \Fbf_2, \Gbf_2, \Hbf_2\}$ into Theorem \ref{thm:quad} to construct a quaternary GCM quad $\{\Pbf, \Qbf, \Rbf, \Sbf\}$ of size $2s_1(t_2+t_3) \times 2t_1(s_2+s_3)$.
\end{enumerate}  
\end{proof}
\begin{remark}
  Interestingly, for $s_1 \times s_2$, $s_1 \times s_3$, $t_2 \times t_1$ and $t_3 \times t_1$ being feasible sizes of quaternary GCM pairs as regulated in Corollary \ref{Qsize},  $s_2+s_3$ or $t_2+t_3$ can cover all the positive integers within 1000 except for 799 and 959. By contrast, the set of the feasible lengths of the sequences quad constructed in \cite{craigen2002complex} does not cover $173$ integers within $1000$. The improved result owes to the more flexible matrix size in one dimension, since by Corollary \ref{Qsize}, it is the product of the sizes in all dimensions that must be a quaternary Golay number, rather than the size in each dimension. For example, to cover the number $87$, we only need two GCM pairs of sizes $n\times 9$ and $n\times 78$, respectively, where $n$ is an indeterminate. If $n=1$, there does not exist a quaternary Golay sequence pair of length $9$ since $9$ is not a quaternary Golay number. However, there exists a quaternary GCM pair of size $2\times 9$ since $18$ is a quaternary Golay number. Besides, by Corollary \ref{Qsize}, there must be at least a BBGN in the factorization of $n$ since there are two $BQGNs$ in $9$. To cover the $173$ integers except for $759$ and $959$, it is revealed by computational search that most of them ($146$ in $171$) have a restriction of only one additional BBGN in the factorization of $n$.
\end{remark}

Nevertheless, the size $2s_1(t_2+t_3) \times 2t_1(s_2+s_3)$ is doubled due to the rigorous quasi-symmetric condition. By the following theorem as the second application of Lemma \ref{lem:compromise}, we do not have to double the size in one dimension.
\begin{theorem} \label{thm:compromise}
  Given a polyphase GCA quad $\{\Abf, \Bbf, \Cbf, \Dbf\}$, where the size of $\Abf$ and $\Bbf$ is $s_1\times \cdots\times s_i\times \cdots\times s_r$, and the size of $\Cbf$ and $\Dbf$ is $s_1\times \cdots\times s_i^\prime\times \cdots\times s_r$, and a polyphase GCA pair$\{\Ibf, \Jbf\}$ of size $t_1\times \cdots\times t_r$, suppose
  \bea \label{eq:disjoint_conca}
  \Abf^{\prime} &= \Abf \vert {\bf{0}}_{s_{i}^\prime}, \quad \Bbf^{\prime} &= \Bbf \vert {\bf{0}}_{s_{i}^\prime},\\
  \Cbf^{\prime} &= {\bf{0}}_{s_i} \vert \Cbf, \quad \Dbf^{\prime} &= {\bf{0}}_{s_i} \vert \Dbf,
  \eea
  where $\vert$ denotes concatenating two arrays in the $i$-th dimension, ${\bf{0}}_{s_i}$ $\left({\bf{0}}_{s_i^{\prime}}\right)$ denotes the zero array of the same size as $\Abf$ $\left(\Cbf\right)$. Suppose
  \bea
  \Ebf &= \Abf^{\prime} \otimes \Ibf + \Cbf^{\prime} \otimes \Jbf,\\
  \Fbf &= \Abf^{\prime} \otimes \Jbf^* - \Cbf^{\prime} \otimes \Ibf^*,\\
  \Gbf &= \Bbf^{\prime} \otimes \Ibf + \Dbf^{\prime} \otimes \Jbf,\\
  \Hbf &= \Bbf^{\prime} \otimes \Jbf^* - \Dbf^{\prime} \otimes \Ibf^*.
  \eea
  Then $\{\Ebf, \Fbf, \Gbf, \Hbf\}$ is a polyphase GCA quad of size $s_1t_1\times \cdots\times \left(s_i+s_i^\prime\right)t_i\times \cdots\times s_rt_r$.
\end{theorem}
\begin{proof}
The proof of Theorem \ref{thm:compromise} is similar to the proof of Theorem \ref{thm:baseArray}, except that the GCA quad instead of the GCA pair is disjoint.
\end{proof}

The way of assigning zeros in Theorem \ref{thm:compromise} is simpler than Proposition \ref{prop:interleave} and Proposition \ref{prop:concate_zero}, since there exists only one composition $+$ in each equation of (\ref{compromise}) in Lemma \ref{lem:compromise}.

% By Theorem \ref{thm:compromise}, let $\Ibf$ and $\Jbf$ be trivial 1, then it degenerates into Proposition \ref{prop:concate}. 
\begin{Corollary} 
  There exist quaternary GCM quads of size $s_1t_1 \times (s_2+s_3)t_2$ where $s_1 \times s_2$, $s_1 \times s_3$, $t_1 \times t_2$ are the feasible sizes of quaternary GCM pairs.
\end{Corollary}
Specifically, we may construct a quaternary GCM quad of size $s_1t_1 \times (s_2+s_3)$ where $t_1$ is a BQGN and $s_2+s_3$ can cover all the positive integers within $1000$ but $799$ and $959$, with an additional restriction on $s_1$, as mentioned in the remark following Corollary \ref{double_sum_size}. We note that the restriction of one additional BBGN in the factorization of $s_1$ can be compensated in the sense of gluing the BQGN $t_1$, due to the disjoint structure of (\ref{eq:disjoint_conca}). E.g., we can construct a quaternary GCM quad of size $36\times 87$: since $87=9+78$, by Theorem \ref{thm:compromise} we only need to feed in three quaternary GCM pairs of sizes $12\times 9$, $12\times78$ and $3\times1$ respectively, which exist according to Corollary \ref{Qsize}. Otherwise, without gluing $t_1$, we need a quaternary GCM pair of size $36\times 9$ which does not exist.

To cover the number $799$, 
noting that $799=(2\times8+1)\times(2\times23+1)$,  we first construct two GCM quads of size $1\times 17$ and $1\times47$ from the base sequences $BS(9, 8)$ and $BS(24, 23)$ by Proposition \ref{prop:interleave} respectively. Then feed them into Theorem \ref{thm:quad} to construct an GCM quad of size $1\times799$. Finally we may enlarge the first dimension via Theorem \ref{thm:baseArray}.

Noting that $959=(2\times 3+1)\times (5+2^{2}\times 3\times 11)$, we cover the number $959$ in the following steps:
\begin{enumerate}
  \item Construct two quaternary GCM pairs $\{\Abf_1, \Bbf_1\}$ and $\{\Cbf_1, \Dbf_1\}$ of size $1\times5$ and $1\times 132$ respectively.
  \item Feed $\{\Abf_1, \Bbf_1, \Cbf_1, \Dbf_1\}$ and a quaternary GCM pair $\{\Ibf, \Jbf\}$ of size $t_1\times 1$ where $t_1$ is a quaternary Golay number, into Theorem \ref{thm:compromise} to construct a quaternary GCM quad $\{\Abf_2, \Bbf_2, \Cbf_2, \Dbf_2\}$ of size $t_1\times 137$.
  \item Feed $\{\Abf_2, \Bbf_2, \Cbf_2, \Dbf_2\}$ into Proposition \ref{prop:concate_zero} to construct an GCM quad $\{\Ebf, \Fbf, \Gbf, \Hbf\}$ of size $4t_1\times 137$.
  \item Construct an GCM quad $\{\Abf, \Bbf, \Cbf, \Dbf\}$ of size $1\times 7$ from $BS(4, 3)$ by Proposition \ref{prop:interleave}.
  \item Feed $\{\Abf, \Bbf, \Cbf, \Dbf\}$ and $\{\Ebf, \Fbf, \Gbf, \Hbf\}$ into Theorem \ref{thm:quad} to construct a quaternary GCM quad $\{\Pbf, \Qbf, \Rbf, \Sbf\}$ of size $4t_1\times 959$.
\end{enumerate}

Thus we have the following corollary:
\begin{Corollary} \label{quaternary_GCM_quad_size}
  There exist quaternary GCM quads whose sizes in one dimension can cover all the positive integers within 1000.
\end{Corollary}
\section{Conclusions} \label{SEC5}
In the sight of four identities over a commutative ring, we propose some construction methods of the polyphase GCA pairs and the polyphase GCA quads.

The feasible sizes of the quaternary GCA pairs are $s_1\times \cdots\times s_r$, where $s_1\cdots s_r$ is of form $2^{a+u}3^{b}5^{c}11^{d}13^{e}$, $a, b, c, d, e, u \geq 0$, $b+c+d+e \leq a+2u+1$, $u \leq c+e$, and each of $u$ factors of $10, 26$ should not be factorized into different dimensions.

The feasible sizes of the binary GCM quads are of form $g_b^{(1)}m\times g_b^{(2)}n$ or $g_b^{(1)}\times g_b^{(2)}mn$ where $m, n\in \{2s+1\vert 0\leq s\leq 38, or \ s=g_b^{(3)}\} \cup \{2(g_b^{(4)}+g_b^{(5)})\}$ and $g_b^{(1)}, \cdots, g_b^{(5)}$ are binary Golay numbers. Specifically, all the sizes within $78\times 78$ can be covered. We conjecture that the binary GCM quads exist for arbitrary sizes.

For the quaternary GCM quad of size $s_1\times s_2$, all the positive integers within 1000 can be covered for $s_1$ or $s_2$. We also obtain the quaternary GCM quads of size $2s_1(t_2+t_3) \times 2t_1(s_2+s_3)$ where $s_1 \times s_2$, $s_1 \times s_3$, $t_2 \times t_1$ and $t_3 \times t_1$ are the feasible sizes of quaternary GCM pairs.

Our Matlab\textsuperscript{TM} codes that can generate GCMs of known feasible sizes in Corollary \ref{Qsize}, \ref{binary_GCM_quad_size} and \ref{quaternary_GCM_quad_size} are available online: \url{https://github.com/csrlab-fudan/ACM}.
\bibliographystyle{IEEEtran}
\bibliography{gca}

\begin{thebibliography}{10}
\providecommand{\url}[1]{#1}
\csname url@samestyle\endcsname
\providecommand{\newblock}{\relax}
\providecommand{\bibinfo}[2]{#2}
\providecommand{\BIBentrySTDinterwordspacing}{\spaceskip=0pt\relax}
\providecommand{\BIBentryALTinterwordstretchfactor}{4}
\providecommand{\BIBentryALTinterwordspacing}{\spaceskip=\fontdimen2\font plus
\BIBentryALTinterwordstretchfactor\fontdimen3\font minus
  \fontdimen4\font\relax}
\providecommand{\BIBforeignlanguage}[2]{{%
\expandafter\ifx\csname l@#1\endcsname\relax
\typeout{** WARNING: IEEEtran.bst: No hyphenation pattern has been}%
\typeout{** loaded for the language `#1'. Using the pattern for}%
\typeout{** the default language instead.}%
\else
\language=\csname l@#1\endcsname
\fi
#2}}
\providecommand{\BIBdecl}{\relax}
\BIBdecl

\bibitem{golay1961complementary}
M.~Golay, ``Complementary series,'' \emph{IRE transactions on information
  theory}, vol.~7, no.~2, pp. 82--87, 1961.

\bibitem{golay1962note}
------, ``Note on'complementary series','' \emph{Proc. IRE}, vol.~50, p.~84,
  1962.

\bibitem{turyn1974hadamard}
R.~J. Turyn, ``Hadamard matrices, baumert-hall units, four-symbol sequences,
  pulse compression, and surface wave encodings,'' \emph{Journal of
  Combinatorial Theory, Series A}, vol.~16, no.~3, pp. 313--333, 1974.

\bibitem{borwein2004complete}
P.~Borwein and R.~Ferguson, ``A complete description of golay pairs for lengths
  up to 100,'' \emph{Mathematics of computation}, vol.~73, no. 246, pp.
  967--985, 2004.

\bibitem{tseng1972complementary}
C.-C. Tseng and C.~Liu, ``Complementary sets of sequences,'' \emph{IEEE
  Transactions on Information theory}, vol.~18, no.~5, pp. 644--652, 1972.

\bibitem{goethals1970skew}
J.~Goethals and J.~Seidel, ``A skew hadamard matrix of order 36,''
  \emph{Journal of the Australian Mathematical Society}, vol.~11, no.~3, pp.
  343--344, 1970.

\bibitem{sivaswamy1978multiphase}
R.~Sivaswamy, ``Multiphase complementary codes,'' \emph{IEEE Transactions on
  Information theory}, vol.~24, no.~5, pp. 546--552, 1978.

\bibitem{frank1980polyphase}
R.~Frank, ``Polyphase complementary codes,'' \emph{IEEE Transactions on
  Information theory}, vol.~26, no.~6, pp. 641--647, 1980.

\bibitem{craigen2002complex}
R.~Craigen, W.~Holzmann, and H.~Kharaghani, ``Complex golay sequences:
  structure and applications,'' \emph{Discrete mathematics}, vol. 252, no. 1-3,
  pp. 73--89, 2002.

\bibitem{bright2021complex}
C.~Bright, I.~Kotsireas, A.~Heinle, and V.~Ganesh, ``Complex golay pairs up to
  length 28: A search via computer algebra and programmatic sat,''
  \emph{Journal of Symbolic Computation}, vol. 102, pp. 153--172, 2021.

\bibitem{luke1985sets}
H.~D. Luke, ``Sets of one and higher dimensional welti codes and complementary
  codes,'' \emph{IEEE Transactions on Aerospace and Electronic Systems}, no.~2,
  pp. 170--179, 1985.

\bibitem{dymond1992barker}
M.~Dymond, ``Barker arrays: existence, generalization and alternatives,''
  \emph{PhD Thesis, University of London}, 1992.

\bibitem{jedwab2007golay}
J.~Jedwab and M.~G. Parker, ``Golay complementary array pairs,'' \emph{Designs,
  Codes and Cryptography}, vol.~44, no.~1, pp. 209--216, 2007.

\bibitem{fiedler2008multi}
F.~Fiedler, J.~Jedwab, and M.~G. Parker, ``A multi-dimensional approach to the
  construction and enumeration of golay complementary sequences,''
  \emph{Journal of Combinatorial Theory, Series A}, vol. 115, no.~5, pp.
  753--776, 2008.

\bibitem{parker2011generalised}
M.~G. Parker and C.~Riera, ``Generalised complementary arrays,'' in \emph{IMA
  International Conference on Cryptography and Coding}.\hskip 1em plus 0.5em
  minus 0.4em\relax Springer, 2011, pp. 41--60.

\bibitem{jiang2019autocorrelation}
Y.~Jiang, F.~Li, X.~Wang, and J.~Li, ``Autocorrelation complementary
  matrices,'' in \emph{2019 53rd Asilomar Conference on Signals, Systems, and
  Computers}.\hskip 1em plus 0.5em minus 0.4em\relax IEEE, 2019, pp.
  1596--1600.

\bibitem{li2021construction}
F.~Li, Y.~Jiang, C.~Du, and X.~Wang, ``Construction of golay complementary
  matrices and its applications to mimo omnidirectional transmission,''
  \emph{IEEE Transactions on Signal Processing}, vol.~69, pp. 2100--2113, 2021.

\bibitem{avis20103}
A.~A. Avis, ``3-phase golay triads,'' Ph.D. dissertation, Science: Department
  of Mathematics, 2010.

\bibitem{avis2021three}
A.~A. Avis and J.~Jedwab, ``Three-phase golay sequence and array triads,''
  \emph{Journal of Combinatorial Theory, Series A}, vol. 180, p. 105422, 2021.

\bibitem{girnyk2021efficient}
M.~A. Girnyk and S.~O. Petersson, ``Efficient cell-specific beamforming for
  large antenna arrays,'' \emph{IEEE Transactions on Communications}, vol.~69,
  no.~12, pp. 8429--8442, 2021.

\bibitem{davis1999peak}
J.~A. Davis and J.~Jedwab, ``Peak-to-mean power control in ofdm, golay
  complementary sequences, and reed-muller codes,'' \emph{IEEE Transactions on
  information theory}, vol.~45, no.~7, pp. 2397--2417, 1999.

\bibitem{paterson2000generalized}
K.~G. Paterson, ``Generalized reed-muller codes and power control in ofdm
  modulation,'' \emph{IEEE Transactions on Information Theory}, vol.~46, no.~1,
  pp. 104--120, 2000.

\bibitem{schmidt2007complementary}
K.-U. Schmidt, ``Complementary sets, generalized reed--muller codes, and power
  control for ofdm,'' \emph{IEEE Transactions on Information Theory}, vol.~53,
  no.~2, pp. 808--814, 2007.

\bibitem{wang2021new}
Z.~Wang, D.~Ma, G.~Gong, and E.~Xue, ``New construction of complementary
  sequence (or array) sets and complete complementary codes,'' \emph{IEEE
  Transactions on Information Theory}, 2021.

\bibitem{jacobson2012basic}
N.~Jacobson, \emph{Basic algebra I}.\hskip 1em plus 0.5em minus 0.4em\relax
  Courier Corporation, 2012.

\bibitem{eliahou1991golay}
S.~Eliahou, M.~Kervaire, and B.~Saffari, ``On golay polynomial pairs,''
  \emph{Advances in applied mathematics}, vol.~12, no.~3, pp. 235--292, 1991.

\bibitem{budivsin2014paraunitary}
S.~Budi{\v{s}}in and P.~Spasojevi{\'c}, ``Paraunitary generation/correlation of
  qam complementary sequence pairs,'' \emph{Cryptography and Communications},
  vol.~6, no.~1, pp. 59--102, 2014.

\bibitem{budivsin2018paraunitary}
S.~Z. Budi{\v{s}}in and P.~Spasojevi{\'c}, ``Paraunitary-based boolean
  generator for qam complementary sequences of length {$2^K$},'' \emph{IEEE
  Transactions on Information Theory}, vol.~64, no.~8, pp. 5938--5956, 2018.

\bibitem{yang1989composition}
C.~Yang, ``On composition of four-symbol $\delta$-codes and hadamard
  matrices,'' \emph{Proceedings of the American Mathematical Society}, vol.
  107, no.~3, pp. 763--776, 1989.

\bibitem{baez2002octonions}
J.~Baez, ``The octonions,'' \emph{Bulletin of the American Mathematical
  Society}, vol.~39, no.~2, pp. 145--205, 2002.

\bibitem{Tarokh1999Space}
V.~Tarokh, H.~Jafarkhani, and A.~R. Calderbank, ``Space-time block coding for
  wireless communications: {P}erformance results,'' \emph{IEEE J. Sel. Areas
  Commun.}, vol.~17, no.~3, pp. 451--460, Mar. 1999.

\bibitem{dokovic1998aperiodic}
D.~Dokovic, ``Aperiodic complementary quadruples of binary sequences,''
  \emph{Journal of Combinatorial Mathematics and Combinatorial Computing},
  vol.~27, pp. 3--32, 1998.

\bibitem{djokovic2010base}
D.~{\v{Z}}. {\DJ}okovi{\'c}, ``On the base sequence conjecture,''
  \emph{Discrete Mathematics}, vol. 310, no. 13-14, pp. 1956--1964, 2010.

\end{thebibliography}

\end{document}